\documentclass[letterpaper,twocolumn,10pt]{article}
\usepackage{usenix-2020-09}
\usepackage{tikz}
\usepackage{adjustbox}
\usepackage{amsmath, amsthm, amssymb}
\usepackage[utf8]{inputenc}
\usepackage{url}
\usepackage{algorithm}
\usepackage[noend]{algpseudocode}
\usepackage{subfigure}
\usepackage{booktabs}
\usepackage{multirow}
\usepackage{pgfplots}
\usepackage{pgfplotstable}
\usepackage[inline]{enumitem}
\pgfplotsset{compat=newest}
\usepackage{authblk}

\begin{document}
\newcounter{margin} 
\setcounter{margin}{1}
\newcommand{\inlinecomment}[2]{{\color{#1}{[#2]}}}
\newcommand{\new}[1]{{\color{purple}{#1}}}

\theoremstyle{plain}
\newtheorem{thm}{Theorem}[section]
\newtheorem{lem}[thm]{Lemma}
\theoremstyle{definition}
\newtheorem{defn}[thm]{Definition}
\newtheorem{exmp}[thm]{Example}
\newtheorem{prop}[thm]{Proposition}
\newcommand\regCoef{\ensuremath{\beta}}  
\newcommand\lr{\ensuremath{\eta}}  
\newcommand\penExp{\ensuremath{\gamma}}  
\newcommand\range{\ensuremath{\lambda}}  
\newcommand\regRange{\ensuremath{\lambda_{reg}}}  
\newcommand\degree{\ensuremath{n}}  

\newcommand\base{\ensuremath{x}}  
\newcommand\exponent{\ensuremath{k}}  
\newcommand\partyA{\ensuremath{A}}  
\newcommand\partyB{\ensuremath{B}}  
\newcommand\partyaVec{\ensuremath{\mathbf{a}}}  
\newcommand\partybVec{\ensuremath{\mathbf{b}}}  
\newcommand\coefficients{\ensuremath{\alpha}}  

\newcommand\totalprecision{\ensuremath{L}}  

\newcommand\precision{\ensuremath{p}}  
\newcommand\scaleDown{\ensuremath{\bar{s}}}  
\newcommand\shares[1]{[[#1]]}  

\newcommand\polytrick{PILLAR}  
\newcommand\binotrick{ESPN}  
\newcommand\honeytrick{HoneyBadger}  
\newcommand\coinn{COINN}  
\newcommand\gforce{GForce}
\newcommand\cryptgpu{CryptGPU}
\newcommand\cheetah{Cheetah}

\newcommand{\plotheight}{5.9cm}
\newcommand{\plotscatter}[5]{
\begin{tikzpicture}[scale=0.9]
\begin{axis}[
    title={#1},
    xlabel={Inference Time (s)},
    ylabel={Test Accuracy (\%)},
    xmode=log,
    grid=major,
    ymin=#3,
    ymax=#4,
    legend pos={#5},
]

\addplot+[
    scatter, only marks,
    scatter src=explicit symbolic,
    scatter/classes={
        This Work={mark=triangle*,blue, fill=blue, mark size=4.5},
        COINN={mark=square*,blue, fill=blue, mark size=3},
        GForce={mark=diamond*,blue, fill=blue, mark size=4},
        CryptGPU={mark=pentagon*,blue, fill=blue, mark size=4}
    },
    error bars/.cd,
    y dir=both,
    y explicit,
] table [
    x=Time,
    y=Accuracy,
    meta=Work,
    col sep=comma,
] {#2};
\addlegendentry{This Work}
\addlegendentry{COINN~\cite{hussainCOINNCryptoML2021}}
\addlegendentry{GForce~\cite{ngGForceGPUFriendlyOblivious2021}}
\addlegendentry{CryptGPU~\cite{tanCryptGPUFastPrivacyPreserving2021}}

\end{axis}
\end{tikzpicture}
}


\date{}

\title{\Large \bf Fast and Private Inference of Deep Neural Networks by Co-designing Activation Functions}

\newcommand*\samethanks[1][\value{footnote}]{\footnotemark[#1]}
\author[ \ 1]{Abdulrahman Diaa\thanks{Equal contribution}}
\author[ \ 1]{Lucas Fenaux\samethanks}
\author[ \ 1]{Thomas Humphries\samethanks}
\author[1]{Marian Dietz}
\author[1]{Faezeh Ebrahimianghazani}
\author[1]{Bailey Kacsmar}
\author[1]{Xinda Li}
\author[1]{Nils Lukas}
\author[1]{Rasoul Akhavan Mahdavi}
\author[1]{Simon Oya}
\author[1, 2]{Ehsan Amjadian}
\author[1]{Florian Kerschbaum}
\affil[1]{University of Waterloo}
\affil[2]{Royal Bank of Canada}
\affil[ ]{\textit {\{abdulrahman.diaa, lucas.fenaux, thomas.humphries, marian.dietz, f5ebrahi, bkacsmar, xinda.li, nlukas, rasoul.akhavan.mahdavi, simon.oya, ehsan.amjadian, florian.kerschbaum\}@uwaterloo.ca}}

\maketitle

\begin{abstract}
Machine Learning as a Service (MLaaS) is an increasingly popular design where a company with abundant computing resources trains a deep neural network and offers query access for tasks like image classification. The challenge with this design is that MLaaS requires the client to reveal their potentially sensitive queries to the company hosting the model. Multi-party computation (MPC) protects the client's data by allowing encrypted inferences. However, current approaches suffer from prohibitively large inference times. The inference time bottleneck in MPC is the evaluation of non-linear layers such as ReLU activation functions. Motivated by the success of previous work co-designing machine learning and MPC, we develop an activation function co-design. We replace all ReLUs with a polynomial approximation and evaluate them with single-round MPC protocols, which give state-of-the-art inference times in wide-area networks. Furthermore, to address the accuracy issues previously encountered with polynomial activations, we propose a novel training algorithm that gives accuracy competitive with plaintext models. Our evaluation shows between $3$ and $110\times$ speedups in inference time on large models with up to $23$ million parameters while maintaining competitive inference accuracy.

\end{abstract}

\section{Introduction}

\begin{figure}[ht]
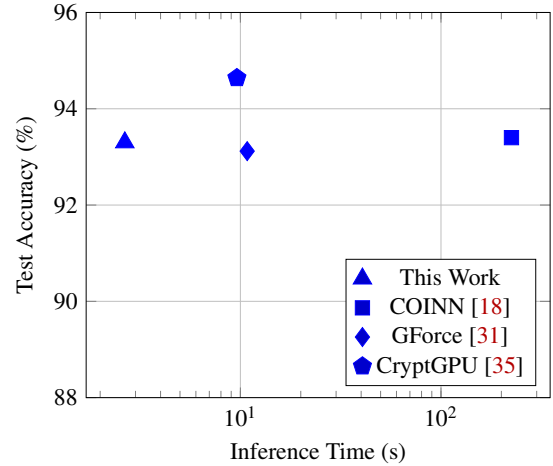

\centering
    \plotscatter{}{csvs/cifar10wan.csv}{88}{96}{south east}
    \caption{Summary of the inference time in seconds vs. test accuracy for each state-of-the-art approach on the CIFAR-10 dataset in the WAN (100 ms roundtrip delay).}\label{fig:cifar_10_summary}
\end{figure}

The rapid development of increasingly capable machine learning (ML) models has resulted in significant demand for products like machine learning as a service (MLaaS).
In this scenario, big tech companies with vast computing resources train large machine learning models and provide users with  query access.
The major pitfall with MLaaS is that it requires clients to submit potentially sensitive queries to an untrusted entity.
A promising solution to this problem is to employ cryptography to ensure the queries and inferences are hidden from the model owner.
Secure inference is an active field of research with many solutions and different threat models as summarized in a recent SoK~\cite{ngSoKCryptographicNeuralNetwork2023}.
The challenge is that despite recent advances, the inference times are still prohibitively large compared to plaintext inferences.

This work focuses on reducing the runtime of secure inference on image data, under realistic network conditions, while maintaining classification accuracy.
We consider the two-party setting using multi-party computation (MPC), where the server holds the modified ML model, and the client holds the data to query the model.
Recent state-of-the-art works in this space employ various co-design approaches to reduce the inference time~\cite{ngSoKCryptographicNeuralNetwork2023}.
For example, \coinn~co-designs ML models optimized for quantization with efficient MPC protocols tailored to the custom models~\cite{hussainCOINNCryptoML2021}.
\coinn~substantially compresses the model and makes numerous optimizations to the architecture to achieve fast inferences.
Another example is GForce, which tailors the cryptography needed for ML to high-speed GPU hardware~\cite{ngGForceGPUFriendlyOblivious2021}.
By offloading vast amounts of work to the pre-computation phase, they are able to achieve state-of-the-art runtime and accuracy in secure inference~\cite{ngGForceGPUFriendlyOblivious2021}.
Similarly, \cryptgpu~\cite{tanCryptGPUFastPrivacyPreserving2021} modifies the CrypTen framework~\cite{knottCrypTenSecureMultiParty2021} to run efficiently on the GPU and give state-of-the-art inference times in wide area networks.
However, despite making major steps towards practical inference, none of these works remove a crucial bottleneck in secure inference: the non-linear layers.

It is well known that the non-linear layers are the bottleneck of secure inference~\cite{garimellaSisyphusCautionaryTale2021,hussainCOINNCryptoML2021,mishraDelphiCryptographicInference2020,CryptoNASProceedings34th}.
This is because secure computation on arithmetic shares is optimized for multiplications and additions, instead of non-linear layers such as ReLU activation functions or MaxPool layers.
In order to compute these non-linear functions, expensive conversions between different types of MPC protocols are required.
Specifically, in more realistic network settings with high latency, the inference time is substantially degraded due to each conversion taking many rounds of communication.
This problem is particularly prevalent in deep neural networks (DNNs), where a non-linear activation separates each of the many linear layers.

This work addresses the non-linear layers by taking a co-design approach between the activation functions and MPC.
We take the approach of replacing classic ReLU activation functions with a polynomial approximation during training and inference to avoid conversions altogether.
Previous work has considered this approach but with limited success~\cite{garimellaSisyphusCautionaryTale2021}. 
We propose two modifications to make this approach practical.
First, we develop and evaluate new single-round MPC protocols that give the fastest evaluation of polynomials to date.
The challenge with using polynomials is they severely impact model accuracy~\cite{hussainCOINNCryptoML2021,mishraDelphiCryptographicInference2020,garimellaSisyphusCautionaryTale2021}.
Previous work could not successfully train DNNs with more than $11$ layers due to exploding gradients~\cite{garimellaSisyphusCautionaryTale2021}.
Thus, our second contribution is tailoring the ML training process to ensure high accuracy and stable training using polynomials.
Our approach utilizes a new type of regularization that focuses on keeping the input to each activation function within a small range.
We achieve close to plaintext accuracy on models as deep as ResNet-110~\cite{heDeepResidualLearning2016} and as large as a ResNet-50 on ImageNet~\cite{dengImageNetLargescaleHierarchical2009}~(23 million parameters).
The combination of these approaches yields a co-design with state-of-the-art inference times and the highest accuracy for polynomial models.

We compare our work with three solutions representing the state-of-the-art approaches in secure inference according to Ng and Chow~\cite{ngSoKCryptographicNeuralNetwork2023}.
We summarize our results in Figure~\ref{fig:cifar_10_summary}.
Combining the single-round MPC protocols with our activation regularization achieves significantly faster inference times than all other solutions.
Specifically, our solution is faster than \cryptgpu~by $4\times$, \gforce~by $5\times$, and \coinn~by $40\times$ on average in wide area networks.
Our approach also scales to large models on ImageNet with a $110\times$ speedup over COINN, $14\times$ speedup over \cheetah, and a $3\times$ speedup over \cryptgpu.
Furthermore, our inference accuracy remains competitive with all other solutions.
\cryptgpu~often gives slightly higher accuracy as it can evaluate any plaintext model (albeit slower than our work).
Thus, the challenge for future work is to further close the ML accuracy gap  between plain and polynomial models.


\section{Background}\label{sec:background}
\subsection{Multi Party Computation}
Secure multi-party computation (MPC) allows a set of parties to jointly compute a function while keeping their inputs to the function private. 
We focus on a variant of MPC which performs operations over shares of the data \cite{ben-orCompletenessTheoremsNoncryptographic1988}.
We use $\shares{s} = [[s]_\partyA, [s]_\partyB]$ to denote that the value of $s$ is shared among participants, where $[s]_\partyA$ is the share held by party $\partyA$ and $[s]_\partyB$ by party $\partyB$. 
Arithmetic MPC protocols utilize a linear secret sharing scheme, such as an additive secret sharing scheme to compute complex circuits using combinations of additions and multiplications.
Given constants $v_1,v_2,v_3$ and shares of values $\shares{x},\shares{y}$, one can locally compute 
\begin{equation}\label{eq:linear_shares}
 v_1\shares{x}+v_2\shares{y}+v_3 = \shares{v_1\cdot x + v_2\cdot y + v_3}
\end{equation}
to obtain shares of the value $v_1\cdot x + v_2\cdot y + v_3$. 
For multiplication, one can use Beaver's trick to multiply using a single round of communication between parties~\cite{beaver91}.
Specifically, we assume a triplet of random numbers $a,b,c$ (called a Beaver triplet) was generated such that $a\cdot b=c$ and secret shared among all parties ahead of time (typically in an offline precomputation phase). 
Then the parties compute $\shares{x\cdot y}$ by first locally computing $\shares{a+x}= A$ and $\shares{b+y} =B$ and reconstructing $A$ and $B$ so that both parties have them in plaintext. 
This reconstruction is the single round required. Using these values, the parties compute the result locally as $\shares{x\cdot y} = A \shares{y} + (-B)\shares{a} +\shares{c}$, using the linearity property in equation~\ref{eq:linear_shares}.

Arithmetic MPC protocols are limited to basic multiplications and additions. Thus, for computing non-linear operations such as comparisons, other techniques such as converting to binary secret shares or using Yao's garbled circuits are common~\cite{demmlerABYFrameworkEfficient2015}. A binary secret sharing scheme is an arithmetic scheme carried out bitwise in the ring $\mathbb{Z}_2$. Specifically, the difference is that we first decompose $x$ into its bits and have a separate arithmetic share of each bit. By maintaining this bitwise structure, operations such as XOR or bit shifts are trivial. We describe how to use a binary and arithmetic secret-sharing scheme together to compute non-linear functions in Section~\ref{sec:motivation}.

\subsection{Neural Network Inference}
We consider DNN classifiers with domain $\mathcal{X}\subseteq \mathbb{R}^d$ and range $\mathcal{Y}\subseteq \mathbb{R}^c$. 
DNN classifiers consist of a sequence of layers, each performing either a linear or a non-linear operation.
The ResNet~\cite{heDeepResidualLearning2016} architecture we consider is composed of (i) convolutional, (ii) fully connected, (iii) pooling, (iv) batch normalization, and (v) ReLU layers.  
All layers are linear, except for $\text{ReLU}(x)= \text{max}(x, 0)$ and max pooling (that can be replaced with average pooling). 
To classify an input $x$, a classifier $h$, passes the input sequentially through each layer. 
Upon reaching the last layer, the prediction is obtained by taking $\text{arg max}_{i\in\{1..c\}}h(x)_i$, where we call $h$ the logit function for a classifier $h: \mathcal{X} \rightarrow \mathcal{Y}$. 
Our work is tailored to securely perform the inference phase.
Specifically, for an encrypted input $x$ we compute the encrypted output $h(x)$ of the logit function. 
However, to achieve improvements during inference, our work requires modification to the training phase that generates $h$.


\section{Problem Setup and Motivation}\label{sec:problem}


\subsection{Problem Setup}
We follow the same threat model as prior work for two-party secure inference~\cite{hussainCOINNCryptoML2021,knottCrypTenSecureMultiParty2021,mishraDelphiCryptographicInference2020,ngGForceGPUFriendlyOblivious2021}.
Specifically, we follow the two-party client-server model where the server has a machine learning model (a DNN) they have trained (with our technique), and the client holds private data upon which they would like to make an inference.
Like Delphi~\cite{mishraDelphiCryptographicInference2020}, our work assumes the server uses a modified training procedure. 
The server's input to the protocol is the weights of their trained model, which they do not want to leak to the client (due to intellectual property or protecting their MLaaS business~\cite{StealingMachineLearning}). 
The client has a private input (typically an image) they would like to classify using the model but do not want to leak this input or the prediction to the server. 
That is, the MPC function can be written as $f(\text{image}, \text{model})= (label, \emptyset)$. 
Following previous work~\cite{hussainCOINNCryptoML2021,knottCrypTenSecureMultiParty2021,mishraDelphiCryptographicInference2020,ngGForceGPUFriendlyOblivious2021}, we consider the semi-honest model~\cite[\S7.2.2]{goldreich09}, where adversaries do not deviate from the protocol but may gather information to infer private information.
Also, in line with previous work, we assume the model architecture is known to both parties. 
This includes the dimensions and type of each layer, parameters such as field size used for inference, and the mean and standard deviation of the training set~\cite{knottCrypTenSecureMultiParty2021}. 

Our goal is to reduce the inference time as much as possible.
We are willing to incur a small degradation in accuracy to achieve practical runtimes.
For example, in the streaming setting, applications like spam detection are runtime-critical (a small accuracy trade-off can be tolerated to make it feasible). 
Since all the protocols we compare with contain pre-computation, we focus on the online phase for fair comparison. 
Furthermore, the online time determines the latency, which we focus on reducing in this work. 
We build from CrypTen~\cite{knottCrypTenSecureMultiParty2021}, which does not implement pre-computation and rather assumes a third-party dealer for Beaver triplets. 
However, in practice, the server and client could generate the Beaver triplets in a pre-computation phase using off-the-shelf protocols~\cite{keller_2016, Rathee_2019}.

\subsection{Privacy During Model Training}
We focus only on the inference phase of machine learning.
However, the privacy of the training process and training data is an orthogonal but essential problem.
We recommend that the data owner take appropriate steps to protect the privacy of the model, such as training using differential privacy~\cite{abadiDeepLearningDifferential2016} or rounding the output of the inference.
Furthermore, during training, care should be taken to protect against threats such as model stealing, which can be launched using only the inference result~\cite{jagielskiHighAccuracyHigh2020}.
To summarize, the model owner learns nothing other than the fact a query was made.
We ensure only the inference is revealed to the client; however, ML attacks that only require black box query access~\cite{jagielskiHighAccuracyHigh2020} must be defended against during the training process.
Since we focus the effects of MPC on the runtime and accuracy of ML inference, we did not conflate this comparison with additional privacy preserving training goals.
Any privacy-preserving technique would add a similar overhead (e.g., reducing the accuracy) to all approaches we evaluate.

\subsection{Motivating the  Co-Design of Activation Functions}\label{sec:motivation}
It has been well established in the literature that activation functions such as ReLU are the bottleneck in MPC-based secure inference, taking up to 93\% of the inference time~\cite{garimellaSisyphusCautionaryTale2021,hussainCOINNCryptoML2021,mishraDelphiCryptographicInference2020,CryptoNASProceedings34th}. 
The reason for this is that current approaches use different types of MPC protocols for a model's linear and non-linear layers~\cite{hussainCOINNCryptoML2021, knottCrypTenSecureMultiParty2021, ngGForceGPUFriendlyOblivious2021,jhaDeepReDuceReLUReduction2021,mishraDelphiCryptographicInference2020}. 
The linear layers are typically computed using standard arithmetic secret-sharing protocols tailored for additions and multiplications.
The non-linear layers are computed using garbled circuits or binary secret share-based protocols.
The bottleneck in wide area networks is typically the conversions between these protocols as they require a large number of communication rounds.
A typical DNN architecture has many linear layers, each followed by a non-linear layer resulting in a prohibitively large number of conversions.

Consider CrypTen, a PyTorch-based secure ML library, as a baseline approach~\cite{knottCrypTenSecureMultiParty2021}.
CrypTen uses binary shares to evaluate boolean non-linear layers such as ReLUs and MaxPooling layers.
Specifically, all linear layers are computed using standard multiplication and addition protocols over arithmetic shares.
To compute $\shares{ReLU(\base)}$ at each layer, $\shares{\base}$ is first converted to binary shares using a carry look-ahead adder.
Once in binary shares, CrypTen extracts the sign bit to compute $\shares{\base>0}$ (a local operation). 
The sign bit, $\shares{\base>0}$, is then converted back to arithmetic shares (trivial for a single bit) and multiplied with $\shares{\base}$ to get $\shares{ReLU(\base)}$. 
The problem with this approach is that each conversion takes $O(log(\totalprecision))$ communication rounds.
Taking into account the additional round needed for multiplication, we observe nine communication rounds per ReLU in practice (under 64-bit precision).
Recently, more sophisticated MPC protocols have been proposed that reduce the number of rounds needed for comparisons in arithmetic shares~\cite{catrinaImprovedPrimitivesSecure2010, catrinaRoundEfficientProtocolsSecure2018} or reduce the cost of binary share conversions~\cite{escuderoImprovedPrimitivesMPC2020}.
However, even if one were to implement these protocols in CrypTen, the number of rounds needed for non-linear layers would still outweigh the number needed for linear layers.

Motivated by this bottleneck, several works have focused on either reducing the number of ReLUs or replacing ReLUs altogether~\cite{CryptoNASProceedings34th,garimellaSisyphusCautionaryTale2021,jhaDeepReDuceReLUReduction2021,leePreciseApproximationConvolutional2021,mishraDelphiCryptographicInference2020,mohasselSecureMLSystemScalable2017}.
One approach is to approximate each ReLU with a high degree polynomial~\cite{garimellaSisyphusCautionaryTale2021, leePreciseApproximationConvolutional2021}.
The advantage of using polynomials is that polynomials can be computed using arithmetic shares, thus removing the need for expensive conversions and improving the total inference time.
A significant challenge with polynomials is maintaining model accuracy~\cite{garimellaSisyphusCautionaryTale2021,hussainCOINNCryptoML2021,CryptoNASProceedings34th,jhaDeepReDuceReLUReduction2021,mishraDelphiCryptographicInference2020}.

Thus, this work aims to provide a secure inference protocol with state-of-the-art inference time and accuracy in realistic networks with high latency.
To do this, we take a co-design approach to balance accuracy and fast inference time.
In Section~\ref{sec:crypto}, we develop MPC protocols that achieve the fastest evaluation of polynomials to date, assuming a modified ML architecture.
In Section~\ref{sec:ml}, we tailor the ML training procedure to achieve high accuracy using this modified architecture.

\section{Faster Evaluation of Polynomials}\label{sec:crypto}
In this section, we evaluate the speed-up of replacing ReLU's with a naive polynomial approximation.
We then develop our single-round protocols and show that they drastically reduce the activation function evaluation time in wide area networks.
\subsection{The Polynomial Advantage}

\begin{figure}
    \centering
    \includegraphics[width=0.75\columnwidth]{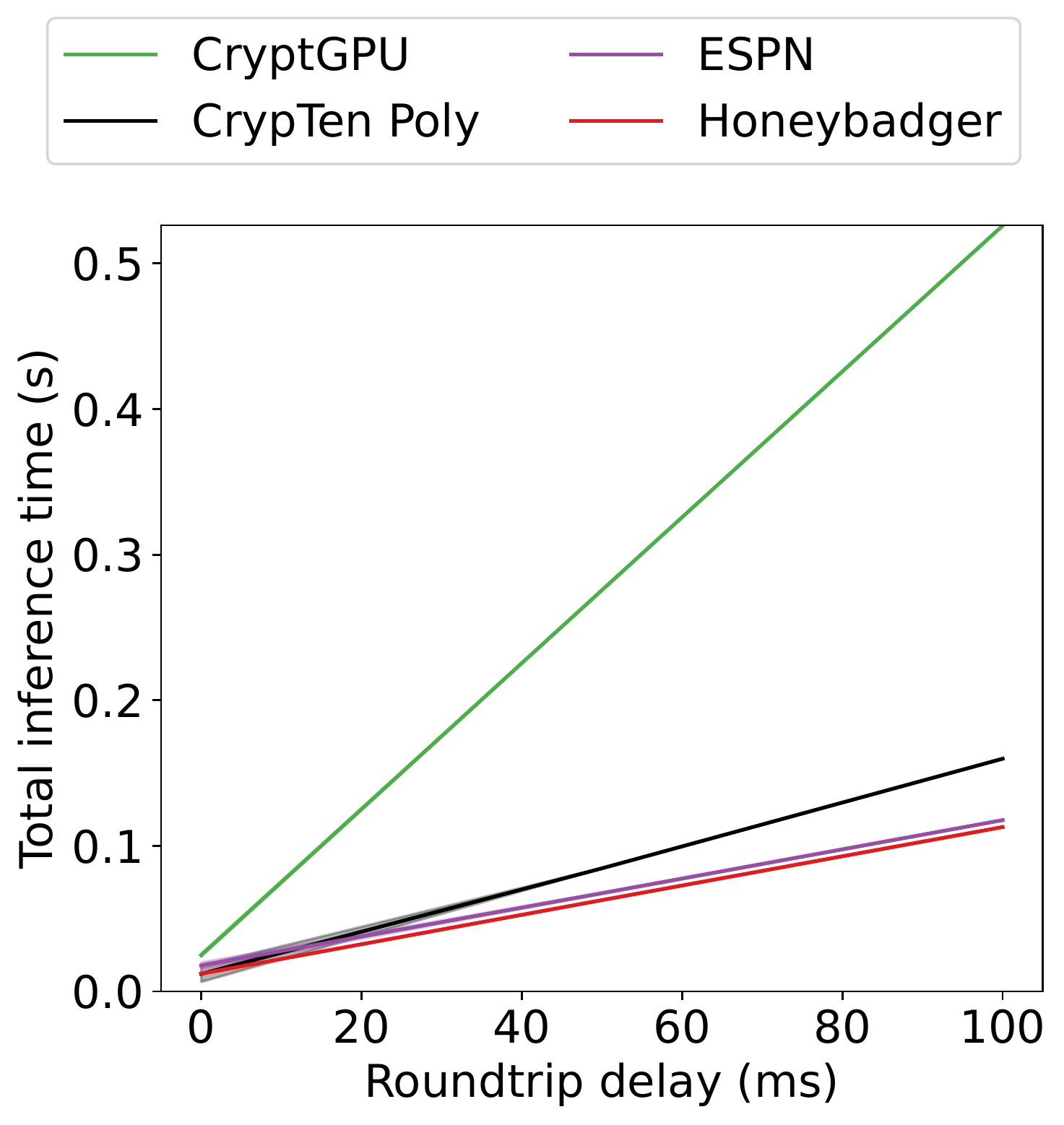}
    \caption{Benchmarking the secure evaluation of ReLU activation functions using various approaches. The $x$-axis is the network delay in ms and the $y$-axis is the mean runtime in seconds averaged over $20$ runs with the shaded area representing the $95\%$ confidence intervals.}
    \label{fig:relu_bench}
\end{figure}

To highlight the speed-up of polynomials over standard ReLUs, we first evaluate the runtime of a single layer with $2^{15}$ ReLU activation functions in Figure~\ref{fig:relu_bench}. (See Section~\ref{sec:experiments} for more implementation details.)
First, we plot an unmodified version of CrypTen (using CryptGPU~\cite{tanCryptGPUFastPrivacyPreserving2021}) with the conversion to binary shares. 
Next, we replace the ReLU with a degree four polynomial fitted using least squares polynomial regression (see Section~\ref{sec:ml} for the details).
We can see that using polynomials in off-the-shelf CrypTen is much faster across all network speeds than the default mixed arithmetic and binary protocol.
This difference becomes more pronounced as we add more network delay or scale to deeper models with more ReLUs.

Despite the significant speed-up, naively computing a polynomial is still expensive in MPC with a non-trivial number of communication rounds.
For example, Horner's method (an iterative approach to evaluating polynomials) uses $O(\degree)$ communication rounds (where $\degree$ is the degree of the polynomial).
However, most MPC libraries (including CrypTen) use the square-and-multiply algorithm for exponentiation, followed by multiplying and summing the coefficients locally.
The square-and-multiply algorithm requires $O(log(\degree))$ multiplications (and thus rounds) in MPC.
In practice, the default square and multiply implementation in CrypTen uses two rounds per ReLU for a degree four polynomial.
To increase the advantage of using polynomial activation functions even further, we develop a new single-round protocol for evaluating polynomials in MPC.

\subsection{\binotrick: \underline{E}xponentiating \underline{S}ecret Shared Values using \underline{P}ascal's tria\underline{n}gle}
We present our single-round, highly parallelizable protocol 
\binotrick~for computing high-degree polynomials.
The fundamental idea is utilizing the binomial theorem (Pascal's triangle) to achieve faster exponentiation.
We begin by describing our protocol for raising a number $\shares{\base}$ to the power $\exponent$, in MPC (see Algorithm~\ref{alg:binomial_exp} for an overview).
Using the additive secret sharing scheme, the exponentiation corresponds to $(\base_\partyA+\base_\partyB)^\exponent$ where $\base_\partyA$ represents the first party's share and $\base_\partyB$ represents the second (such that $\base_\partyA + \base_\partyB= \base$).
The binomial theorem expands this expression as:
\begin{equation}\label{eq:bino_formula}
 \base^\exponent = (\base_\partyA +\base_\partyB)^\exponent = \sum_{i=0}^\exponent {\exponent \choose i} \base_\partyA ^{\exponent-i} \base_\partyB^i
\end{equation}
We observe that, for each $i$ in the sum, party $\partyA$ can compute $\partyaVec_i = \base_\partyA ^{\exponent-i}$ without needing to communicate with party $\partyB$ (Alg.~\ref{alg:binomial_exp} line~\ref{line:partya}).
Similarly, party $\partyB$ can compute $\base_\partyB^i$ without communicating with party $\partyA$ (Alg.~\ref{alg:binomial_exp} line~\ref{line:partyb}).
Finally, $\exponent \choose i$ can be computed by any party (or pre-computed ahead of time).
For simplicity, we assign the computation of $\exponent \choose i$ to party $\partyB$.
Thus party $\partyB$ computes $\partybVec_i = {\exponent \choose i} \base_\partyB^i$. 

Once each party has computed their respective vectors, we multiply $\partyaVec_i \cdot \partybVec_i$ for each $i$ in parallel (Alg.~\ref{alg:binomial_exp} line~\ref{line:mult_shares}).
We carry out this multiplication using standard MPC protocols in one round.
To use these multiplication protocols, each party must have a share of the input.
We use a trivial additive secret sharing, where the other party inputs zero as their share to the protocol (Alg.~\ref{alg:binomial_exp} line~\ref{line:init}).
Finally, after the multiplication, the sum of the binomial theorem can be efficiently computed with no communication (Alg.~\ref{alg:binomial_exp} line~\ref{line:get_result}).

\begin{algorithm}
\caption{Exponentiation Protocol for 2-party additive secret-sharing}\label{alg:exp}
\begin{algorithmic}[1]
 \Procedure{$Exp$}{$\shares{\base}, \exponent$}
 \State $\mathbf{\partyaVec} = \shares{\mathbf{0}^\exponent}$, $\mathbf{\partybVec} = \shares{\mathbf{0}^\exponent}$ \Comment{Initialize shares}\label{line:init}

 \For{$i = 0:\exponent$}
 \State Party $\partyA$ computes: $\partyaVec_i = \partyaVec_i +\base_\partyA^{\exponent-i}$\label{line:partya}
 \State Party $\partyB$ computes: $\partybVec_i =\partybVec_i + {\exponent \choose i} \base_\partyB^i$ \label{line:partyb}
 \EndFor
 \State $\mathbf{p} = BeaverMultiply(\mathbf{\partyaVec}, \mathbf{\partybVec})$\Comment{Parallel Multiplication}\label{line:mult_shares}
 \State $s = \sum_{i=0}^\exponent \mathbf{p_i}$ \Comment{$s = \mathbf{\partyaVec} \cdot \mathbf{\partybVec}$}\label{line:get_result}
 \State \textbf{return} $s$\Comment{secret-shares of $\base^\exponent$}
 \EndProcedure
\end{algorithmic}\label{alg:binomial_exp}
\end{algorithm}
\subsection{Polynomial Evaluation with ESPN}\label{sec.poly_eval}
\paragraph{Floating Point Considerations.}
Our initial description of \binotrick considers the integer domain for simplicity. 
Extending to floating point values is straightforward but requires rescaling (a standard practice in fixed-point arithmetic).
We use CrypTen's two-party, local truncation protocol to ensure we do not incur additional rounds.
However, there is a negligible chance of an incorrect result from this truncation protocol due to wrap-around in the ring.
Specifically, the probability of an incorrect result when truncating $x$ is $\frac{|x|}{2^\totalprecision}$ where $2^\totalprecision$ is the size of the ring~\cite[Appendix C.1.1]{knottCrypTenSecureMultiParty2021}. 
This implies that $x$ must be small compared to the ring for this fast truncation protocol to be correct.

\paragraph{Polynomial Evaluation Protocol.}
Using our exponentiation protocol, we show how to compute high-degree polynomials efficiently in a single round.
We overview the protocol in Algorithm~\ref{alg:poly_eval}.
First, in parallel, we compute all needed exponents using Algorithm~\ref{alg:binomial_exp}. 
We recall that to ensure the correctness of truncation, we must ensure all intermediate values remain small.
For simplicity, we define small to be that no intermediate scale becomes larger than twice the working precision $\precision$.
We show the complete failure probability calculations of Algorithm~\ref{alg:poly_eval} in Appendix~\ref{app:correct_proofs}.

To ensure the values remain small after exponentiation, we create multiple scaled-down copies of the input $\base$, proportional to each exponent we need to calculate. 
To compute $\base^i$ we first scale $\base$ down by $2^{-\scaleDown}$ where $\scaleDown = \lceil (i-2) \precision / i \rceil)$ and $\precision$ is the current working precision of $\base$ (line~\ref{line:pre_scale_down}).
$\scaleDown$ is chosen such that $2^\precision$ becomes approximately $2^{2\precision/i}$ after scaling and thus approximately $2^{2\precision}$ after exponentiation by $i$.
These are approximations as not all values of $i$ divide $(i-2) \precision$, so there is some error from taking the ceiling. 
To account for the additional factor of approximately two, we do an additional rescaling after each exponentiation in line~\ref{line:post_scale_down}.
This rescaling incorporates $\scaleDown$ (which includes the ceiling function) to ensure all values are scaled back to $2^\precision$.
Finally, after computing all powers, we can locally multiply the result by the coefficients (public values) and sum in line~\ref{line:get_poly}.
We give a complete proof of correctness for Algorithm~\ref{alg:binomial_exp} and \ref{alg:poly_eval}, including the truncation operator in Appendix~\ref{app:correct_proofs}.

\begin{algorithm}
    \caption{Polynomial Evaluation Protocol for 2-party additive secret-sharing}\label{alg:exp}
    \begin{algorithmic}[1]
     \Procedure{$Poly$}{$\shares{\base}, \coefficients, \degree$}
     \State $\mathbf{p}_1 = \base$
     \For{$i = 2:\degree$} \Comment{In parallel}
     \State Let $\scaleDown = \lceil (i-2) \precision / i \rceil$ \Comment{Scale down factor}\label{line:scale_down_factor}
     \State $x'_i = x \cdot 2^{-\scaleDown}$\Comment{Scale down before Exp}\label{line:pre_scale_down}
     \State $\mathbf{p}_i = Exp(x'_i,i)$\Comment{From Algorithm~\ref{alg:binomial_exp}}\label{line:exp}
     \State $\mathbf{p}'_i =  \mathbf{p}_i * 2^{-\precision*(i-1)+ \scaleDown\cdot i}$\Comment{Scale down after Exp}\label{line:post_scale_down}
     \EndFor
     \State $y = \coefficients_0 + \sum_{i=1}^\degree \mathbf\coefficients_i \cdot {p}'_i$ \Comment{Locally dot product}\label{line:get_poly}
     \State \textbf{return} $y$\Comment{Secret-shares of $f(x)$}
     \EndProcedure
    \end{algorithmic}\label{alg:poly_eval}
    \end{algorithm}

    \paragraph{Hyperparameter Restrictions.}
    While Algorithm~\ref{alg:poly_eval} is designed to keep intermediate values small, multiple hyperparameters determine the protocol's effectiveness.
    Using results from Appendix~\ref{app:correct_proofs} (namely the max of Theorem~\ref{thm:first_trunc} and Theorem~\ref{thm:second_trunc}), we get that the probability of failure for a given truncation is bounded by
    \begin{equation}\label{eq:precision}
        Pr[\text{Truncation Failure}] \leq \frac{2^{\degree(\lceil \log_2{\range}\rceil + 1) + 2\precision}}{2^{\totalprecision}}.
    \end{equation}
    
    For our experiments, we use CrypTen, with $\totalprecision=64$~\cite{knottCrypTenSecureMultiParty2021} and a default working precision of $\precision=16$. Assuming default values of $\degree=4$ and $\range=5$ we get a failure probability bound of $2^{-16}$. However, this is a pessimistic upper bound since we consider the worst-case input of $\pm 5$.
    Conversely, in our experiments, we find that the distribution of inputs follows an approximately Laplace distribution as shown in Appendix~\ref{app:poly_fit}. 
    Thus, we observe a much smaller empirical failure probability.
    
    While our default parameters were experimentally chosen to give high classification accuracy (Section~\ref{sec:experiments}, shows a minor degradation over plaintext accuracy), it is unclear if these values are optimal.
    For example, one approach to reducing the failure probability is to decrease $\precision$ from $16$ to $10$-bit (which gives a failure probability of $2^{-28}$).
    However, the trade-off is that the intermediate $\base'_i$ values will be truncated severely.
    For instance, with $\degree=4$, $x'_i$ will be truncated to $5$-bits. 
    We find that this loss of precision is too significant to simulate a ReLU function accurately. 
    Similarly, while a higher degree polynomial might better approximate a ReLU, a higher degree will negatively affect both the failure probability and the precision loss of $x'_i$.
    Future work could conduct an extensive hyperparameter search of all parameters to find the optimal trade-off.
    However, as we see in Section~\ref{sec:experiments} (Tables~\ref{tab:cryptgpu}-\ref{tab:imagenet}), this would \emph{at most} yield a $0.5\%$ increase in encrypted classification accuracy (PyTorch vs. CrypTen accuracy), for the worst model and dataset.

\paragraph*{Evaluating Algorithm~\ref{alg:poly_eval}.}
In Figure~\ref{fig:relu_bench}, we plot this approach alongside the previous approaches to evaluate the runtime.
\binotrick~incurs slightly more overhead in the LAN setting; however, it scales significantly better (the confidence intervals do not overlap) to wide area networks that can be expected in practice.

\subsection{Alternative Single Round Protocol: HoneyBadger}

Like \binotrick, Lu et al. give a single round protocol for exponentiation in MPC~\cite{luHoneyBadgerMPCAsynchroMixPractical2019}.
Despite focusing on a completely different problem (anonymous communication), they provide an MPC protocol of independent interest for exponentiation, which we also utilize in our work.
They take a very different approach to our work that yields different trade-offs.
Instead of the binomial theorem, their work utilizes the following factoring rule
\begin{equation}\label{eq:honeybadger}
 \base^\exponent - r^\exponent = (\base-r)\sum\limits_{i=0}^{\exponent-1} \base^{\exponent -i -1}r^i
\end{equation}
where $r$ is a random secret-shared number derived during pre-computation.
We assume each party has a share of $\base$ and a share of $r^i$ for $i\in\{1,\dots,\exponent\}$ before beginning the protocol (instead of the more common Beaver triplets).
The first step in the protocol is to compute and reveal $\base-r$ ($\base$ blinded by $r$), which uses a single round.
Once revealed, this value becomes a public constant $C$.
After some algebraic manipulation of (\ref{eq:honeybadger}), Lu et al. obtain a recursive formula for $\base^ir^j$ given below.
\begin{equation}\label{eq:honey2}
 \shares{\base^\exponent r^j} = \shares{r^{\exponent+j}} + C\sum\limits_{i=0}^{\exponent-1} \shares{\base^{\exponent -i -1}r^{i+j}}
\end{equation}
Using dynamic programming, the parties can then compute any power ($\base^\exponent r^0$) using only additions of previously computed terms and powers of $r$.
To compute polynomials using this protocol, we simply swap the call to $Exp$ in line~\ref{line:exp} of Algorithm~\ref{alg:poly_eval}.

The advantage of Lu et al.'s protocol is that the communication is small (only the opening of $\base-r$). 
The primary disadvantage is that the protocol requires a modified pre-computation phase, which is as difficult to pre-compute securely as the original problem (it is exponentiation).
On the contrary, our binomial protocol uses standard Beaver triplets commonly found in MPC frameworks. 
There are well established protocols for efficiently computing these triplets, and the parties may already have them due to the popularity of Beaver's trick.
A more minor disadvantage of Lu et al.'s solution is that, while the protocol requires very little communication, it is not locally parallelizable as each dynamic programming step depends on the previous one. 
In contrast, our entire protocol can be executed in parallel. 

We also consider the runtime of using HoneyBadger in Figure~\ref{fig:relu_bench}.
We emphasize this is a runtime-only evaluation. 
Without our training algorithm in Section~\ref{sec:ml}, none of the polynomial-based solutions can attain usable accuracy.
We find that \binotrick~and HoneyBadger perform similarly in practice, with HoneyBadger gaining a slight advantage in very low network delay. 
Due to the pre-computation trade-offs, we will evaluate both approaches for the remainder of this work.

\begin{figure*}[htbp]
\centering

\begin{minipage}[t]{0.31\textwidth}
    \centering
    \includegraphics[width=\columnwidth]{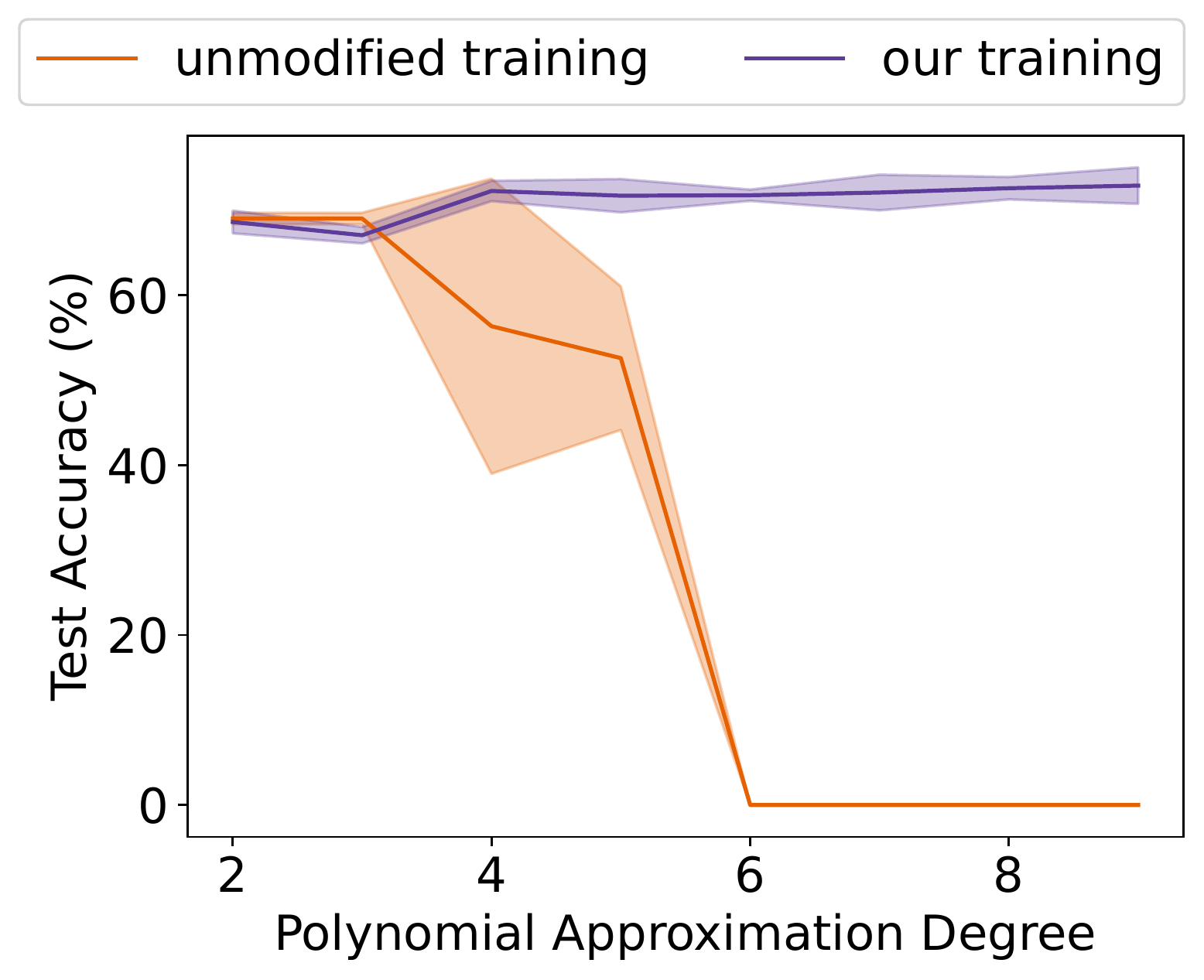}
    \caption{Accuracy of a 2-layer convolutional network trained with varying degrees for the polynomial activation function.} \label{fig:accuracy_vs_degree}
\end{minipage}\hfill%
\begin{minipage}[t]{0.31\textwidth}
    \centering
    \includegraphics[width=\columnwidth]{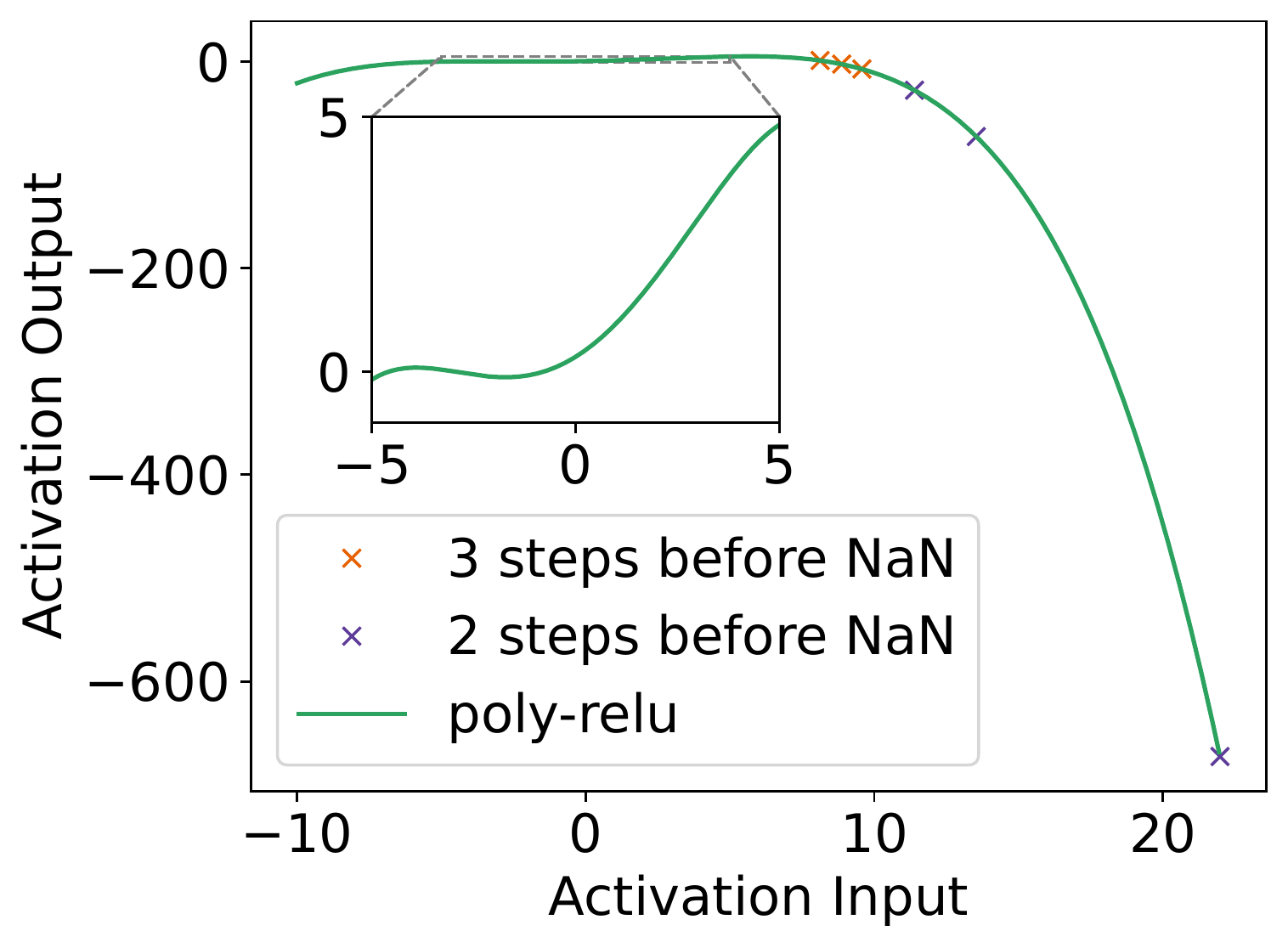}
    \caption{Illustrating the escaping activation problem for the two layers convolutional network.}
    \label{fig:blow_up}
\end{minipage}\hfill%
\begin{minipage}[t]{0.31\textwidth}
    \centering
    \includegraphics[width=\columnwidth]{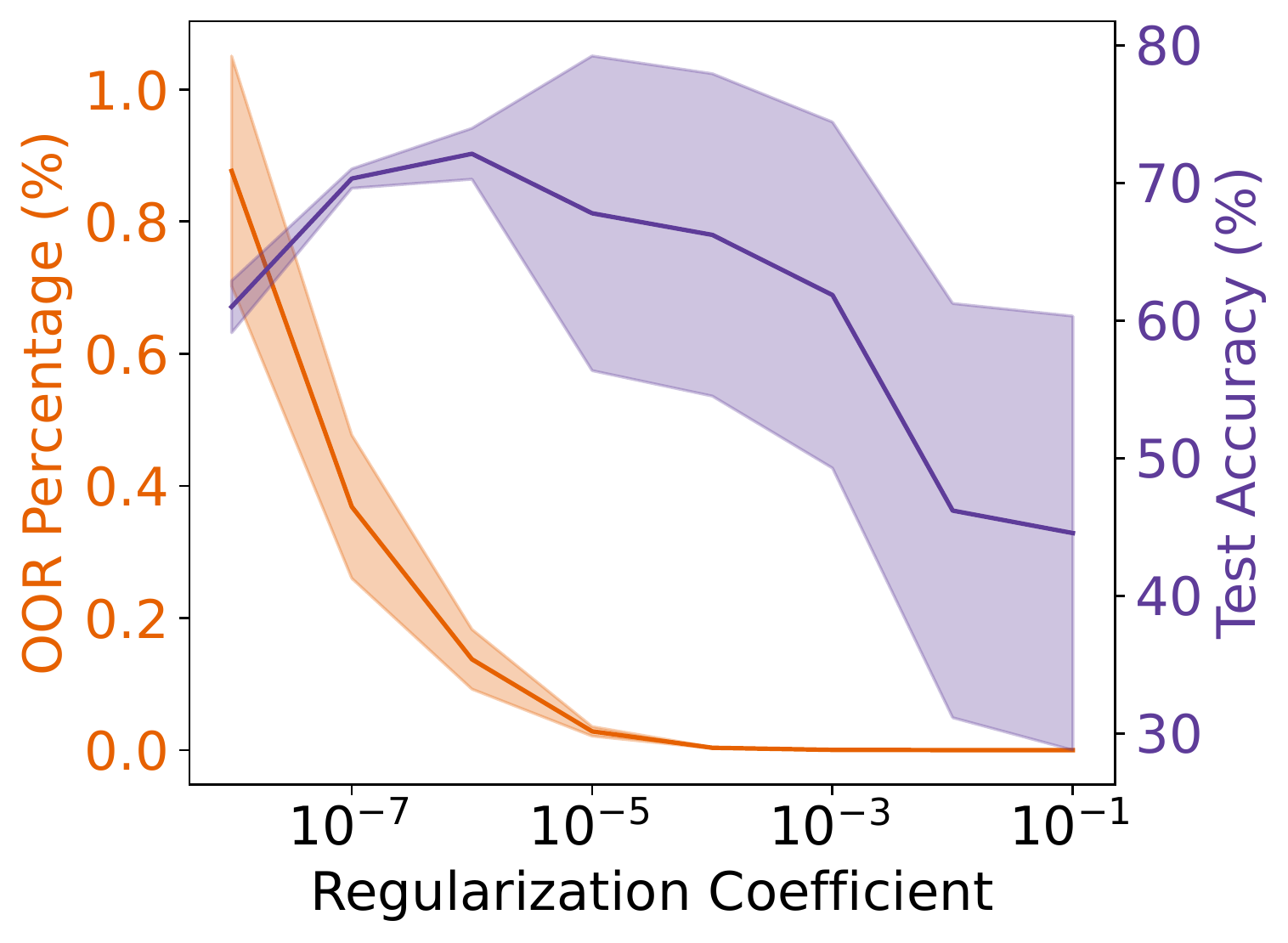}
    \caption{The effect of  the regularization coefficient, $\regCoef$ on model accuracy and out-of-range ratio for $\penExp=10$}\label{fig:OOR_coeff}

\end{minipage}

\end{figure*}

\section[\polytrick:~\underline{P}olynom\underline{i}a\underline{l}~\underline{A}ctivation \underline{R}egularization]{\polytrick:~\underline{P}olynom\underline{i}a\underline{l}~\underline{A}ctivation \\  \underline{R}egularization}\label{sec:ml}
Our initial benchmark in Section~\ref{sec:crypto} showed a significant speed-up when replacing ReLU functions with polynomials implemented using ESPN and HoneyBadger. 
However, a notable challenge neglected thus far is that replacing a ReLU with a polynomial can drastically reduce the accuracy of the model~\cite{hussainCOINNCryptoML2021,garimellaSisyphusCautionaryTale2021}. 
This section discusses the causes of the accuracy degradation and  describes our mitigation techniques. 
Finally, we give empirical results showcasing the high accuracy of our modified training procedures across various architectures and datasets.

\subsection{The Problem with Polynomial Activation Functions}\label{sec:poly_blow_up}

\paragraph{Escaping Activations.}
The first step in replacing an activation function with a polynomial is to design a polynomial that closely approximates the original function.
A common approach for this is the least-squares polynomial fitting.
In this approach, a table of values is created for the polynomial over a small discretized range of values.
This creates a system of equations for the polynomial coefficients that can be solved with least squares.
The challenge with this approach is that, outside of this range, the polynomial no longer resembles the original activation function and often diverges rapidly.
This leads to a problem called escaping activations, first identified by Garimella et al.~\cite{garimellaSisyphusCautionaryTale2021}.
If one naively swaps a ReLU for its polynomial approximation, all weights will become infinite within a few training epochs.
We give an example of this degradation in Figure~\ref{fig:accuracy_vs_degree}. 
We can see that, without modifying the training procedure, a polynomial can completely destroy the accuracy.

To illustrate the problem more clearly, we conduct an experiment using a polynomial of degree four fitted on the range $[-5,5]$ (\range=5) as the activation function for a three-layer model on CIFAR-10~\cite{CIFAR}.
In Figure~\ref{fig:blow_up}, we plot this polynomial activation function and the $\ell_\infty$-norm of the input and output to each activation function.
We note that, with no modification (except replacing ReLUs with polynomials), the weights of this model become undefined within approximately three epochs of training. 
First, we note the divergent behaviour of the polynomial outside the fitted range. 
Second, we observe the effect of the divergence on the outputs of the activation function.
Specifically, we wait until the model weights become undefined (NaN in Python) and then observe the behaviour leading up to the explosion.
We can see that three steps before the model weights become undefined (NaN), the input values of each activation are out-of-range, but the outputs still behave similarly to a ReLU.
However, in the next iteration (two steps before NaN), a single value in the first layer goes too far out of range.
This causes a ripple effect for the other two layers, creating an extremely large output (approx $-2000$) in the final activation function.
This large value creates a large gradient, and after another iteration of training, the values become so large that the gradients (and weights) become undefined (NaN).
We find that minimizing the classification loss alone is not enough to keep the model in range as the gradients explode before decreasing the loss.

\paragraph{Truncated Polynomial Coefficients.}
An additional challenge is that we will evaluate the fitted polynomial in a finite ring with limited precision.
This significantly impacts the polynomial coefficients, which tend to be relatively small, especially for the higher-order terms. 
Specifically, these small coefficients can get truncated to zero in limited precision, which causes the polynomial to diverge even inside the fitted range.
We give an example of this in Appendix~\ref{app:poly_fit}.

\subsection{Defining PILLAR}
Our approach, which we call \polytrick, is the combination of the components we describe in this section.
Activation function regularization is our primary approach for mitigating escaping activation functions. 
However, to scale to larger models, we find that the additional steps of clipping, regularization warm-up, and adding batch normalization are beneficial.

\paragraph{Quantization-Aware Polynomial Fitting.}
We begin by solving the problem of truncated polynomial coefficients.
To address this, we fit the polynomial with the precision constraint in mind.
We do this by using mixed integer non-linear programming. 
Let $X$ be the set of all values between $[-\range,\range]$ in $\precision$-bit precision (the domain we want to fit on).
First, we generate $Y=ReLU(X)\cdot 2^{\precision}$, a table of values for a standard ReLU scaled up by the precision.
Scaling the output of the ReLU allows us to work in the integer domain (similar to fixed point arithmetic). 
We then compute a matrix $B$ where each column is the different powers of $X$ used in a polynomial ($B=[X^0, X^1, X^2, \dots]$).

Next, we solve the system $AB=Y$ for $A$ using mixed integer linear programming with $A\in [-2^{\precision}-1, 2^{\precision}-1]$ to get the coefficients $A$ that minimize the error between the polynomial $AB$ and the ReLU values $Y$.
Finally, we scale the resulting coefficients down by $2^{\precision}$.
We note that $A\in [-2^{\precision}-1, 2^{\precision}-1]$ corresponds to coefficients being bounded by $[-1,1]$ after we scale down.
We empirically choose $\precision=10$ for all polynomials during the ML training. 
As we observe in Appendix~\ref{app:poly_fit}, our quantized polynomial fitting addresses the problems of exploding activations within the range. 
However, the issue of going out-of-range requires additional treatment.

\paragraph{Activation Regularization.}
Following the observations of Section~\ref{sec:poly_blow_up} and Garimella et al.~\cite{garimellaSisyphusCautionaryTale2021}, it is clear that minimizing the classification loss alone is not sufficient to prevent escaping activations. 
Garimella et al.~proposed QuaIL, a method that trains one layer of the model at a time, focusing not on classification accuracy but the similarity of the layer to a standard ReLU model~\cite{garimellaSisyphusCautionaryTale2021}. 
QuaIL showed much better accuracy than naive training but only scaled to models with at most $11$ layers.

In our work, we address the cause of the problem directly by regularizing the input to each activation function during training. 
We add an exponential penalty to the loss function when the model inputs out-of-range values to the polynomial activation function.
Let $\base$ be the input to the activation function, and $\regRange$ be the upper bound of the symmetric range $[-\regRange, \regRange]$ in which we would like the input to be contained. 
Then, we define our penalty function as
\begin{equation}\label{eq:penalty_function}
 p(x) = \left(\frac{x}{\regRange}\right)^\penExp
\end{equation}
where $\penExp$ is a large even number (to handle negative values) determining the severity of the penalty. 
We find that values between six and ten work best in practice, with $\penExp=10$ being the default in our experiments. 
This penalty function gives negligible penalties (less than 1) for $|x|<\regRange$ and rapidly grows (in the degree of $\penExp$) as $|x|>\regRange$.

We aggregate $p(x)$ over $\mathcal{I}$, the set of inputs to all activation functions, by taking the average over each activation layer in the model.
After aggregation, we scale the penalty using a regularization coefficient $\regCoef$ and add it to the existing cross-entropy loss function of the model $\ell_c$. 
Specifically, the modified loss function $\ell'$ is defined as:
\begin{equation}\label{eq:loss}
    \ell'(\cdot) = \ell_c(\cdot) + \frac{\regCoef}{K}\sum\limits_{x \in \mathcal{I}} p(x)
\end{equation}
where $K$ is the number of activation layers in the model.
This allows us to tune the importance of classification loss vs.~the cost of going out-of-range. 

\paragraph{Clipping.}
Although activation regularization teaches the model not to go out-of-range over time, the model still needs to avoid going to infinity during the early stages of training. 
Thus, during training, we apply a clipping function to the input of the activation function such that if any input goes out of range, it is truncated to the range's maximum (or minimum) value. 
This clipping function does not affect the penalty as it is applied after the penalty function has been computed. 
We emphasize that this clipping function is only used during training and is removed during inference. 
The intuition is that the model should learn not to go out-of-range during training and thus no longer requires this clipping function during inference. 
Additionally, we find that setting the $\regRange$ of the penalty to be smaller than the range used for polynomial fitting (and clipping) can yield even better results. This is because the polynomial will be accurate for a larger range outside of the range the model was regularized to stay inside, allowing an extra buffer in case of failure during inference.

\begin{table}[t]
\centering
\begin{tabular}{cllcc}
\toprule
\multicolumn{1}{l}{Dataset} & Model     & \multicolumn{2}{c}{Plain Accuracy $\pm$ CI} & \\
\cmidrule(lr){3-4}
                                &           & ReLU & PILLAR \\
\midrule
\multirow{4}{*}{Cifar10}      & MiniONN   & 91.2 $\pm$ 0.17 & 88.1 $\pm$ 0.26 \\
                                & VGG 16    & 92.6 $\pm$ 0.16 & 90.8 $\pm$ 0.11 \\
                                & ResNet18  & 94.7 $\pm$ 0.09 & 93.4 $\pm$ 0.14 \\
                                & ResNet110 & 92.8 $\pm$ 0.27 & 91.4 $\pm$ 0.18 \\
\cmidrule(lr){1-5}
\multirow{3}{*}{CIFAR-100}    & VGG 16    & 70.9 $\pm$ 0.17 & 66.3 $\pm$ 0.22 \\
                                & ResNet32  & 68.4 $\pm$ 0.46 & 67.8 $\pm$ 0.32 \\
                                & ResNet18  & 76.6 $\pm$ 0.07 & 74.9 $\pm$ 0.14 \\
\cmidrule(lr){1-5}
ImageNet                      & ResNet50  & 80.8 & 77.7 \\ \bottomrule
\end{tabular}
\caption{Plain-text Accuracy of \polytrick ~(5 runs).}\label{tab:plain_acc}
\end{table}
\paragraph{Regularization Warm-up.}
We find the minimum requirements for successfully training a model with polynomial activation functions are activation regularization and clipping. 
However, for larger models, the penalty term can be extremely large in the first few epochs (until the model learns to stay in range). 
In some cases, the loss can become infinite due to our regularization penalty. 
To address this challenge, we adopt a regularization scheduler for the first four epochs that slowly increases both $\penExp$ and $\regCoef$ to the values used for the rest of the training. 
Empirically, the following schedule works well and avoids infinite loss. 
We let $\penExp' \in \{4,6,\dots, \penExp, \penExp, \dots\}$ and $\regCoef' \in \{\regCoef/100, \regCoef/50, \regCoef/10, \regCoef/5, \regCoef, \regCoef,\dots \}$.

\paragraph{BatchNorm Layers.}
Garimella et al. also investigated using normalization to help prevent the escaping activation functions~\cite{garimellaSisyphusCautionaryTale2021}. 
They proposed a min-max normalization approach where each layer's minimum and maximum values are approximated using a weighted moving average of the true minimum and maximum. 
These values are frozen during inference. 
Garimella et al. observed that this approach alone was insufficient, as activations still escaped the range during inference. 
We observe this operation is similar to the batch norm layer commonly added to ML models. 
The main difference is that the mean and standard deviation of the batch are used to normalize the layer instead of the minimum and maximum values. 
By fixing the approximation of the mean and standard deviation during inference (following CrypTen~\cite{knottCrypTenSecureMultiParty2021}), this operation is very efficient in MPC. 
We study the effect of BatchNorm in Appendix~\ref{app:batchnorm}.
We find that batch norm layers considerably improve the accuracy of \polytrick.  
This is an intuitive result as batch normalization helps to keep each layer's output bounded and thus reduces the work of our regularization function.

\paragraph{Summary of PILLAR}
We refer to \polytrick as the combination of all components described in this section. We note that the clipping and regularization warm-up components are used (only during training) to enable regularization by preventing the model from going to infinity. The regularization component ensures that the model will stay in range during inference. All three of these components play a crucial role in the success of \polytrick, as removing any of them will result in a model with unacceptable accuracy (due to infinite weights or escaping activations). The batch norm is the only optional component, which we give a small ablation study over in Appendix~\ref{app:batchnorm}.
\subsection{Measuring PILLAR's Effectiveness}

\paragraph{The Regularization Coefficient.}
To show the effect of our regularization and coefficient $\regCoef$, we conduct an experiment using the same three-layer model on CIFAR-10 from Section~\ref{sec:poly_blow_up}. 
In Figure~\ref{fig:OOR_coeff}, the left $y$-axis gives the out-of-range ratio (OOR), defined as the ratio of activation function inputs that were not within the interval $[-5,5]$.
The right $y$-axis is standard classification accuracy, and the $x$-axis varies the regularization coefficient $\regCoef$.
We observe that when the coefficient, $\regCoef$, is small, the model goes out-of-range often and thus has poor accuracy. 
As we increase $\regCoef$, the out-of-range ratio decreases, and accuracy increases. 
However, if we increase the coefficient too much, the accuracy decreases again.

\paragraph{End-to-end Accuracies.}
We evaluate \polytrick~across a range of different models and architectures considered in related work~\cite{hussainCOINNCryptoML2021,ngGForceGPUFriendlyOblivious2021,tanCryptGPUFastPrivacyPreserving2021}. We summarize the results in Table~\ref{tab:plain_acc}. We include the accuracy of a model trained with standard ReLUs as a baseline. All results are averaged over five random seeds, and we show the $95\%$ confidence interval. The only exception is ResNet50 on ImageNet, where we only train a single model due to the size of the dataset. We defer to Section~\ref{sec:experiments} for the details of the experimental setup.
We note that these results are using PyTorch with no cryptography or quantization. We give a complete evaluation using MPC in Section~\ref{sec:experiments} where quantization has an effect. Table~\ref{tab:plain_acc} provides preliminary evidence that our polynomial training approach yields high accuracies competitive with state-of-the-art ReLU models across a range of models and datasets.

\section{Evaluation of Co-Design}\label{sec:experiments}
\begin{figure}[t]
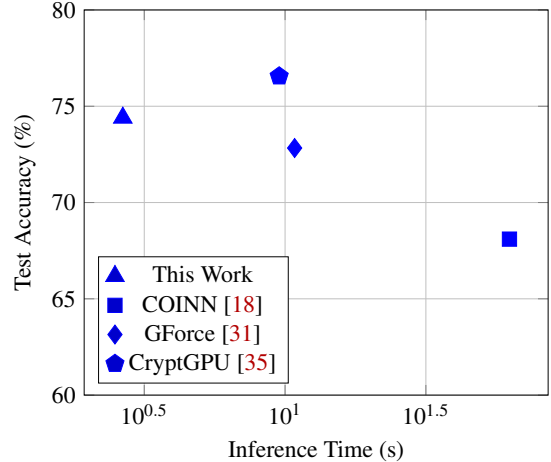

\centering
    \plotscatter{}{csvs/cifar100wan.csv}{60}{80}{south west}
    \caption{Summary of the inference time vs accuracy for each state-of-the-art approach on the CIFAR-100 dataset in the WAN (100 ms roundtrip delay).}\label{fig:cifar_100_summary}
\end{figure}
In this section, we provide an end-to-end comparison of our co-design against state-of-the-art solutions in secure inference. 
We evaluate the performance of \polytrick~and Algorithm~\ref{alg:poly_eval} using both \binotrick~and \honeytrick~ as they offer different trade-offs in the type of pre-computation needed and the size of the communication.
We determine the state-of-the-art works following a recent SoK by Ng and Chow~\cite{ngSoKCryptographicNeuralNetwork2023}.
Specifically, we consider three solutions on the Pareto front of latency and accuracy as determined by Ng and Chow.
These works are COINN~\cite{hussainCOINNCryptoML2021}, GForce~\cite{ngGForceGPUFriendlyOblivious2021} and CrypTen (CryptGPU)~\cite{knottCrypTenSecureMultiParty2021,tanCryptGPUFastPrivacyPreserving2021}. 
We also evaluate Cheetah~\cite{Haung_Cheetah_22}, a recent work not included in the SoK.
We will evaluate the metrics of latency (or runtime of a single sample) and encrypted accuracy.
We additionally evaluate the communication and number of rounds in Appendix~\ref{app:comm_rounds}.
We begin with the experimental setup, then evaluate both metrics (runtime and accuracy) against each related work.
Section~\ref{sec:res_18} evaluates the ResNet-18 architecture, which gives our state-of-the-art performance. Section~\ref {sec:gforce_comp} evaluates the VGG-16 architecture, the only architecture GForce evaluates. Section~\ref {sec:coinn_comp} considers other ResNets and the MiniONN architecture following COINN. Finally, in Section~\ref{sec:imagenet}, we evaluate ImageNet against \cheetah, \coinn, and \cryptgpu.

\paragraph{Results Summary.}

We plot a summary of the accuracy and inference time for CIFAR-100 in Figure~\ref{fig:cifar_100_summary}. For both datasets (recall Figure~\ref{fig:cifar_10_summary} ), we observe that our work always gives the solution with the fastest inference time by a statistically significant amount. In terms of accuracy, our work is competitive with the state-of-the-art, but \cryptgpu~is always the most accurate as it can infer unmodified plaintext models.
Our solution is faster than \cryptgpu~by $4\times$, \gforce~by $5\times$, and \coinn~by $18\times$ on average in wide area networks. 
Our accuracies are competitive with state-of-the-art and plaintext solutions and stay stable (no escaping activations) with models containing up to $110$ layers and $23$ million parameters.

\subsection{Experimental Setup}\label{sec:exp_setup}
We develop an experimental setup that follows as closely as possible to the works we compare to~\cite{hussainCOINNCryptoML2021,ngGForceGPUFriendlyOblivious2021,tanCryptGPUFastPrivacyPreserving2021}. We use CIFAR-10/100~\cite{CIFAR} and ImageNet~\cite{dengImageNetLargescaleHierarchical2009}, the same common benchmark datasets as related work. Our model architectures include: MiniONN~\cite{liuObliviousNeuralNetwork2017}, VGG~\cite{simonyanVeryDeepConvolutional2015}, and ResNets~\cite{heDeepResidualLearning2016}.  This covers models of depth $7$ to $110$ layers with the number of trainable parameters ranging from $0.2$ to $23$ million. 

\paragraph{Implementation Details.}
All experiments are run on a machine with 32~CPU cores @ 3.7~GHz and 1~TB of RAM with two NVIDIA A100 with 80~GB of memory.
We simulate network delay by calling the sleep function for the appropriate time whenever the client and server communicate. We simulate the LAN with 0.25~ms roundtrip delay and the WAN with 100~ms, following COINN~\cite{hussainCOINNCryptoML2021}. We additionally evaluate a real WAN using AWS instances in Section~\ref{sec:real_wan}. All experiments (except ImageNet) are repeated over multiple random seeds, and we report the mean and 95\% confidence interval as shaded areas.

For all related work, we run our own benchmarks of their code unmodified.
While results for ResNet32 and MiniONN appear in the \cheetah~paper, there was no source code for these models so we only evaluate \cheetah~on ImageNet.
To use \cryptgpu~in practice, one must first train a model in PyTorch. For ImageNet, PyTorch provides pre-trained models. However, we will need to train a model for all other architectures and datasets. We simply use the same configurations as our PolyRelu models but with standard ReLUs. We include all source code to reproduce our results ~\cite{code}.

\paragraph{Hyperparameters.}
We introduce five new hyperparameters associated with our techniques: polynomial degree (\degree), polynomial approximation range (\range), polynomial regularization range (\regRange), polynomial regularization coefficient (\regCoef), and polynomial regularization exponent (\penExp). The default values for each are decided by extensive grid searches. These parameters primarily affect the accuracy and not the inference time, except for the polynomial degree (\degree), which has a minor effect on the communication size (but not the rounds).
We found the optimal quantization-aware polynomial degree (\degree) to be 4. We found that higher degrees than four increase the failure probability of truncation with minimal accuracy gains (Section~\ref{sec.poly_eval}). Conversely, lower degrees give lower accuracy (as shown in Figure~\ref{fig:accuracy_vs_degree}). The polynomial used in all evaluations is:  $0.31445312 + 0.5x + 0.15625x^2 - 0.00292969x^4$. 

For all models and datasets, we found a value of $\penExp=10$ performs well as it introduces a strong enough incentive for PolyReLU inputs to stay in range (smaller values did not) while keeping penalization for values in range practically $0$ (if an input $x$ is within range, then $ \frac{x}{\regRange} < 1 \Rightarrow (\frac{x}{(\regRange)^{10}} \sim 0 $). Larger values than this often lead to an infinite penalty term. 
We found that a polynomial approximation range $\range=5$ provides a good compromise between regularization and quantization. We found that smaller values of $\range$ destroy the accuracy during training, and larger values use too much precision, increasing the failure probability. 
The optimal polynomial regularization range (\regRange) and polynomial regularization coefficient (\regCoef) vary per model and dataset, although we found $\regRange=4.8$ (slightly tighter than $\range=5$) and $\regCoef=5\times 10^{-5}$ to be good default values. 
We use $\precision=10$ during ML training, but revert to CrypTen's default precision of $\precision=16$ during inference in a $64$-bit ring.

We used Stochastic Gradient Decent as the optimizer with a learning rate of $0.013$ as the default. This includes a Cosine Annealing Learning Rate Scheduler with Linear Learning Rate Warmup of 5 epochs and decay $0.01$. We use a weight decay of either $10^{-4}$ or $5\times 10^{-4}$ and a momentum of $0.9$. We used a default batch size of $128$ and set the default number of Epochs to $185$. For some models, we tuned the learning rate, number of epochs, and regularization coefficient to achieve a slightly higher accuracy. We detail hyperparameters in our source code repository.
\begin{figure}[t]
 \centering
 \includegraphics[width=0.70\columnwidth]{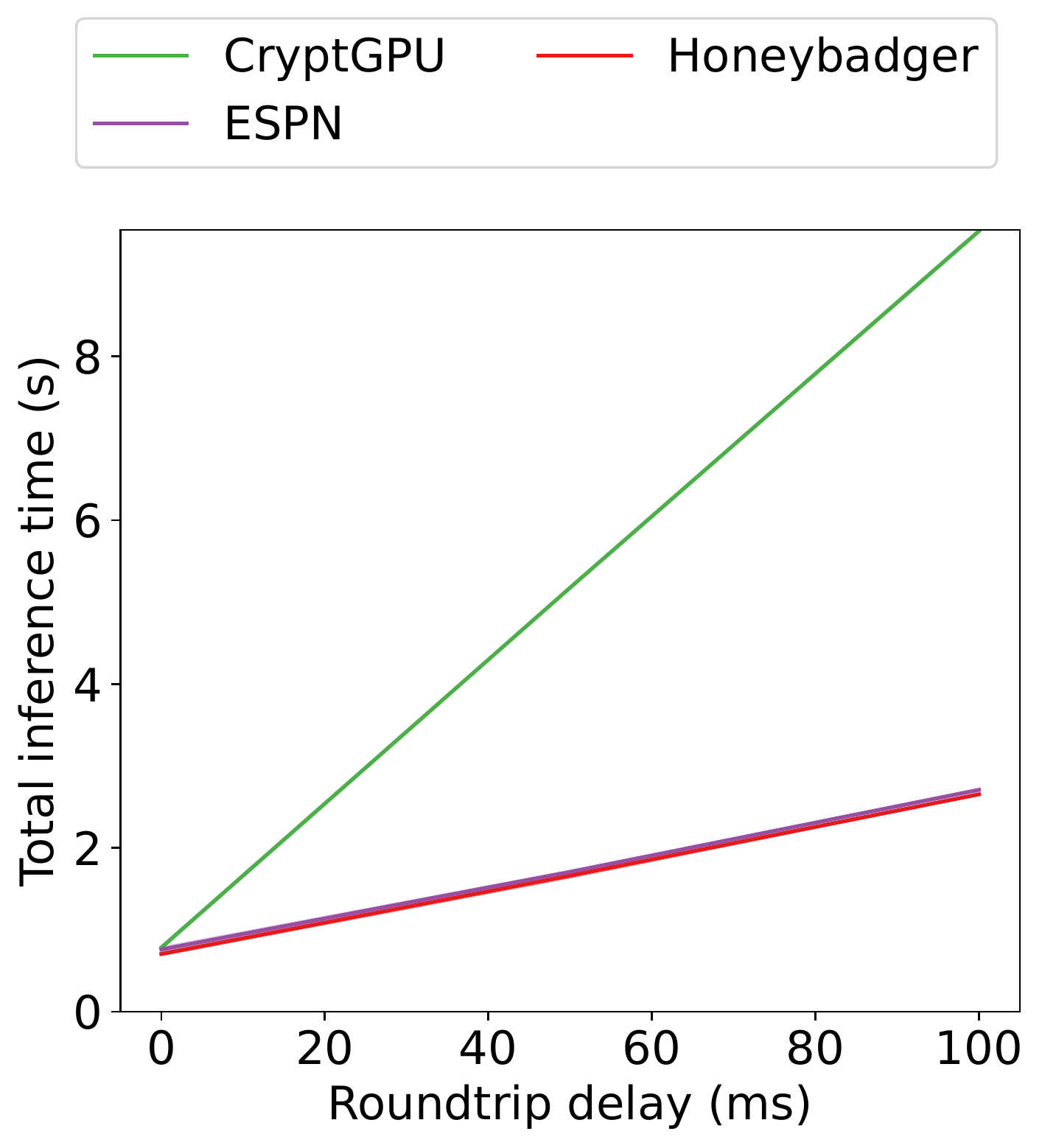}
 \caption{ResNet-18 evaluation on CIFAR-100 (20 runs).}
 \label{fig:CryptGPU}
\end{figure}
\subsection{ResNet-18 Architecture}\label{sec:res_18}
In this section, we use a ResNet-18 architecture as it is the architecture that yields the best inference time and accuracy over all CIFAR-10 and CIFAR-100 experiments. 
For this comparison we focus on \cryptgpu, which has been shown to be a state-of-the-art solution~\cite{ngSoKCryptographicNeuralNetwork2023}.
\cryptgpu~(or CrypTen)~\cite{tanCryptGPUFastPrivacyPreserving2021} serves as a baseline in all our comparisons including those against \gforce~and \coinn~in Sections~\ref{sec:gforce_comp} and~\ref{sec:coinn_comp}.
Neither COINN nor GForce support the ResNet-18 architecture evaluated in this section. 

\paragraph{Inference Time.}
We measure the inference time of a single input image over varying network delays. The results are given in Figure~\ref{fig:CryptGPU}. We include the result for CIFAR-100 and omit the plot for CIFAR-10 as it displays similar trends. We observe that both \polytrick~+ \honeytrick~and \polytrick~+ \binotrick~outperform \cryptgpu~with statistical significance across all roundtrip delays (as the shaded area does not overlap). In the WAN (100 ms), this corresponds to a $4\times$ speedup over \cryptgpu. Furthermore, we find \honeytrick~and \binotrick~perform similarly as observed in Section~\ref{sec:crypto}, with \polytrick~+ \honeytrick~having a slight advantage.

\paragraph{Accuracy.}
We measure the accuracy of the models on the testing set both in plain (using PyTorch) and encrypted (using CrypTen).
We give the result in Table~\ref{tab:cryptgpu}. First, we observe the plain and encrypted accuracies are very similar, indicating that quantization has a minor effect despite not considering this in training. We find that \cryptgpu~and \polytrick~give similar accuracies, with \cryptgpu~performing slightly better, as is to be expected since they use unmodified activation functions. However, we argue this slight loss in accuracy is well justified by the significant decrease in inference time.
\begin{table}[ht!]
    \begin{tabular}{llll}
        \toprule
    Dataset                    & Technique                 & Plain Acc                    & Enc Acc                       \\ \midrule
    \multirow{2}{*}{CIFAR-10}  & \polytrick & 93.4 $\pm$ 0.14 & 93.3 $\pm$ 0.22 \\
                               & \cryptgpu  & 94.7 $\pm$ 0.09 & 94.6 $\pm$ 0.10 \\ \cmidrule(lr){1-4}
    \multirow{2}{*}{CIFAR-100} & \polytrick & 74.9 $\pm$ 0.14 & 74.4 $\pm$ 0.31 \\ 
                               & \cryptgpu  & 76.6 $\pm$ 0.07 & 76.6 $\pm$ 0.13
    \\ \bottomrule
    \end{tabular}
    \caption{ResNet-18 accuracy comparison (5 runs).}\label{tab:cryptgpu}
\end{table}

\begin{figure}[t!]
 \centering
 \includegraphics[width=0.70\columnwidth]{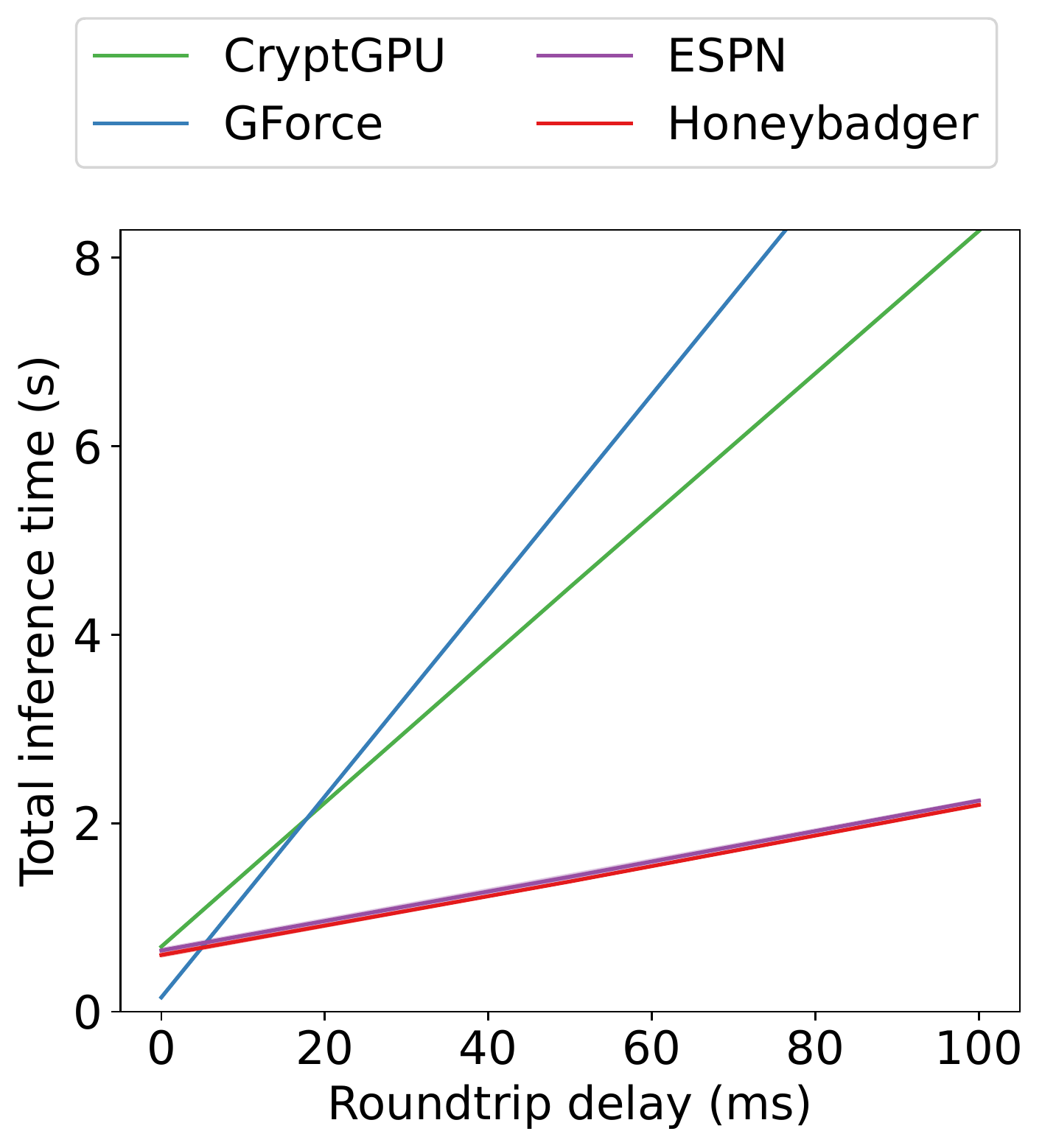}
 \caption{VGG-16 evaluation on CIFAR-100 (20 runs).}
 \label{fig:gforce}
\end{figure}

\begin{figure*}[t!]
 \subfigure[MiniONN on CIFAR-10]{
 \includegraphics[height=\plotheight]{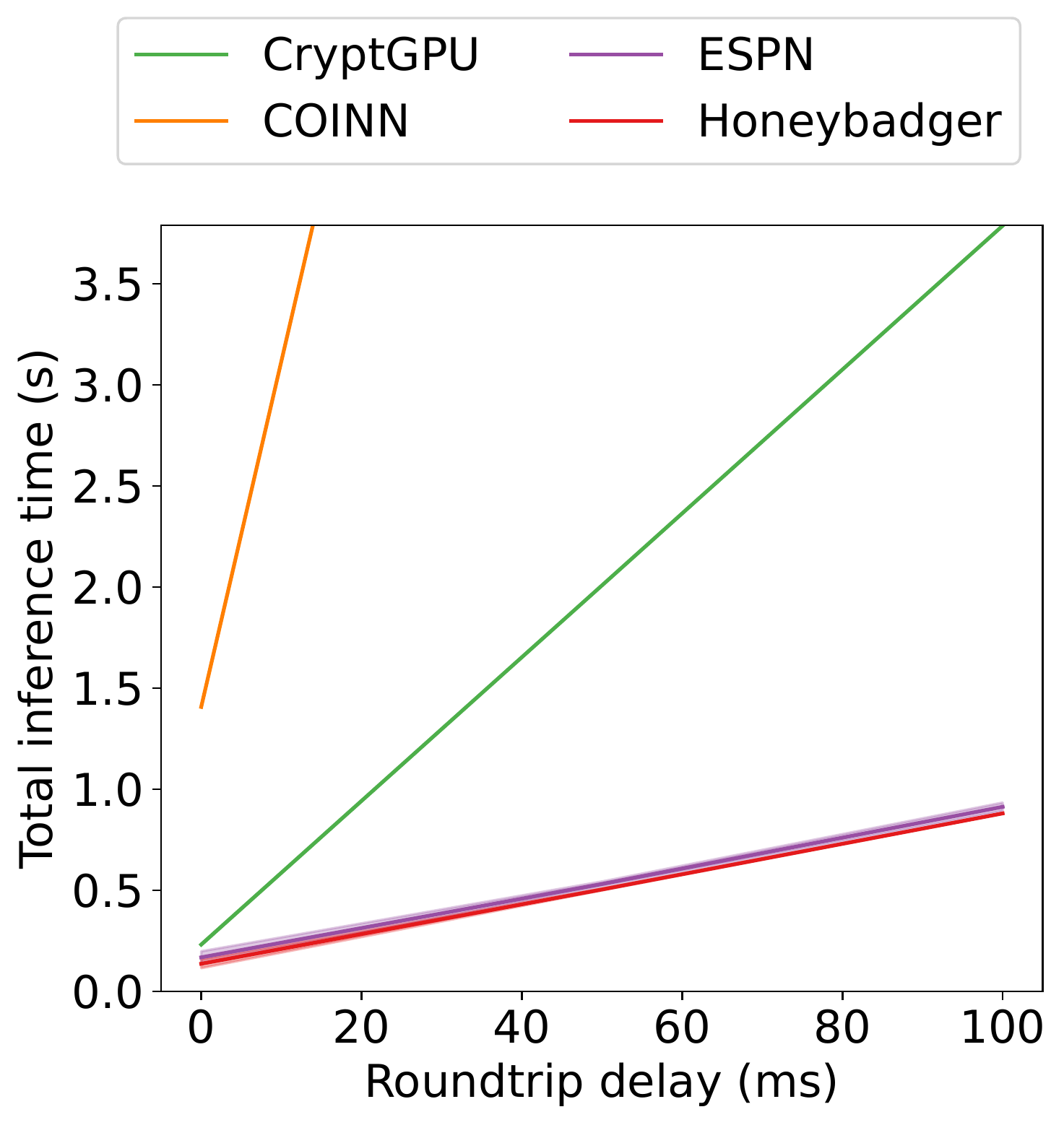}
 \label{fig:coinn_mini}}%
 \hfill
 \subfigure[ResNet-32 on CIFAR-100]{
 \includegraphics[height=\plotheight]{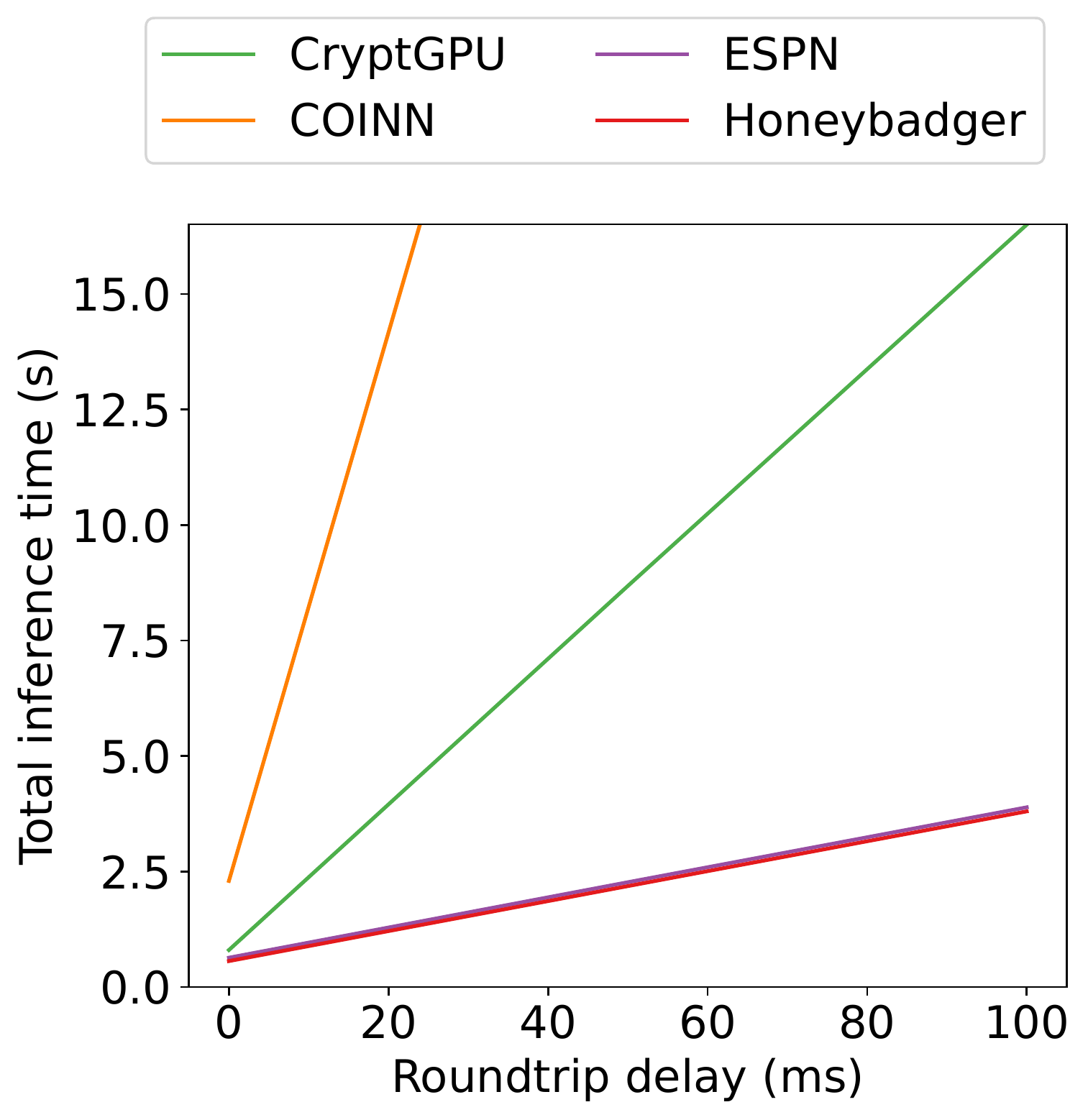}
 \label{fig:coinn_res32}}
 \hfill
 \subfigure[ResNet-110 on CIFAR-10]{
 \includegraphics[height=\plotheight]{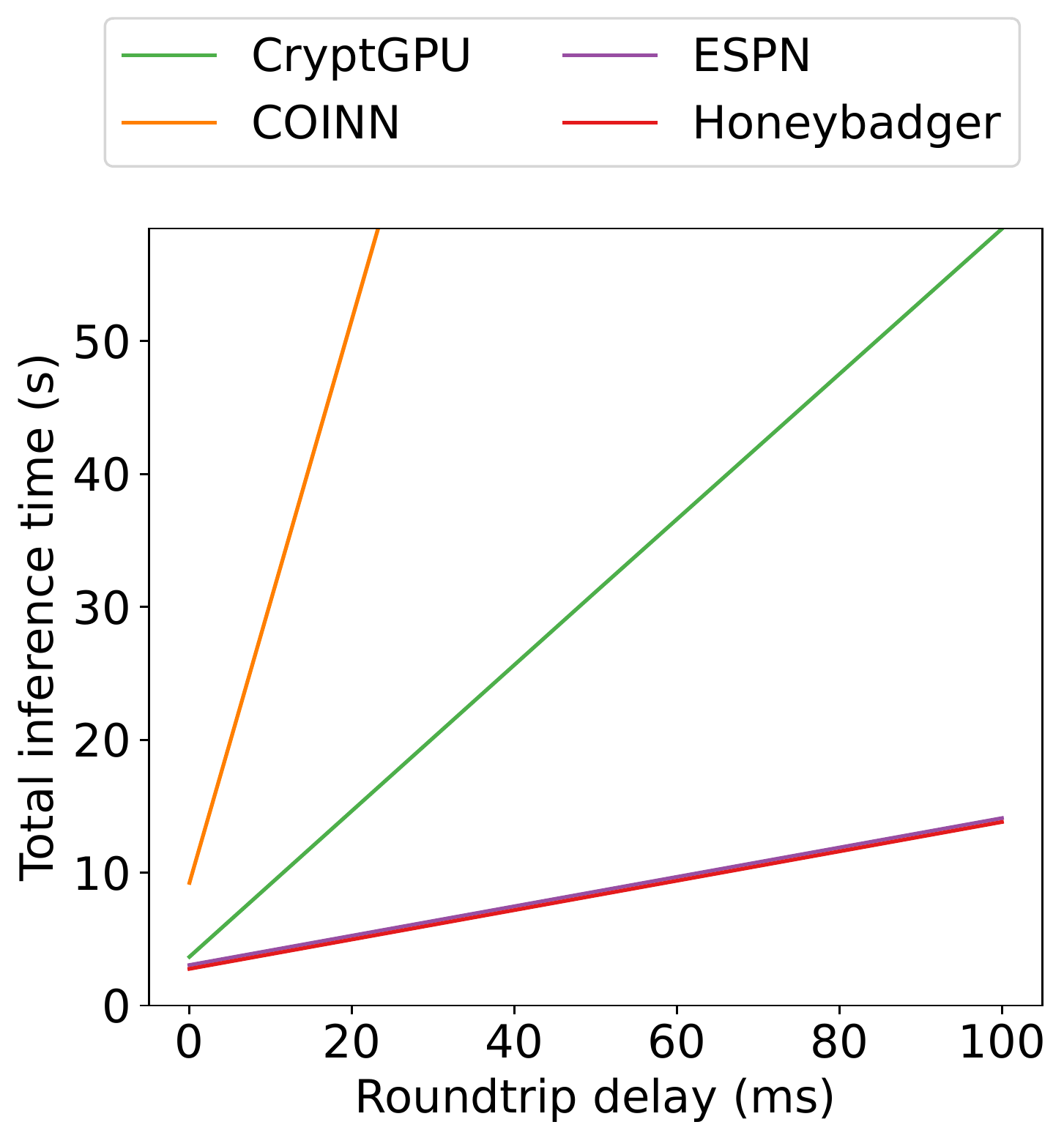}
 \label{fig:coinn_res110}}
 \caption{Evaluating the various COINN architectures (20 runs).}
 \label{fig:coinn}
\end{figure*}
\subsection{VGG-16 Architecture}\label{sec:gforce_comp}
In this section, we compare with GForce, the current state-of-the-art as shown by Ng and Chow~\cite{ngSoKCryptographicNeuralNetwork2023}. GForce focused on a modified VGG-16~\cite{simonyanVeryDeepConvolutional2015} architecture and compared it to all other works (including those using different architectures). For completeness, we evaluate the VGG-16 architecture using our techniques, and \cryptgpu~although we note that the ResNet-18 architecture outperforms VGG-16 in both inference time and accuracy.
COINN~\cite{hussainCOINNCryptoML2021} does not give results for VGG-16, so we exclude it from this section.

\paragraph{Inference Time.}
Since our work aims to reduce the rounds needed by binary non-linear layers, we replace the MaxPool layers in the VGG-16 with AvgPool for all solutions (including \cryptgpu~and GForce). 
We give the inference times over various delays in Figure~\ref{fig:gforce}.
First, we note that GForce significantly outperforms all other solutions in the LAN.
However, for more realistic high latency networks (>5ms roundtrip delay), we observe our solutions significantly outperform GForce ($5\times$ speedup in WAN).
Once again, our solutions outperform \cryptgpu~for all network delays.

\paragraph{Encrypted Accuracy.}
We recall that we swap the MaxPool layers for AvgPool layers in the inference time evaluation. This comes at a cost to accuracy for the VGG architecture. Thus, to give the best scenario possible for GForce, we consider the accuracy of GForce with MaxPools and our work with AvgPool. 
We give the results in Table~\ref{tab:gforce}.
As expected, GForce outperforms our work in accuracy (due to the MaxPools); however, only by a few percentage points.
We train \cryptgpu~to use AvgPool and find it also loses a few percentage points, confirming our hypothesis that a MaxPool is necessary for high accuracy in VGG-16.
We emphasize that our ResNet-18 result outperforms GForce's VGG result in inference time and accuracy.
Furthermore, ResNets are a more popular and compact architecture due to skip-connections.
\begin{table}[H]
    \begin{tabular}{llll}
        \toprule
    Dataset                    & Technique  & Plain Acc & Enc Acc \\
    \midrule
    \multirow{3}{*}{CIFAR-10}  & \polytrick & 90.8 $\pm$ 0.11 & 90.8 $\pm$ 0.14 \\
                               & \cryptgpu  & 92.6 $\pm$ 0.16 & 92.5 $\pm$ 0.16 \\
                               & \gforce    & -         & 93.12       \\\cmidrule(lr){1-4}
    \multirow{3}{*}{CIFAR-100} & \polytrick & 66.3 $\pm$ 0.22 & 66.3 $\pm$ 0.32 \\
                               & \cryptgpu  & 70.9 $\pm$ 0.17 & 70.8 $\pm$ 0.13 \\
                               & \gforce    & -         & 72.83     
    \\ \bottomrule
    \end{tabular}
    \caption{VGG-16 accuracy comparison (5 runs).}\label{tab:gforce}
\end{table}

\subsection{Other Architectures}\label{sec:coinn_comp}
While GForce is the current state-of-the-art, COINN is a competitive solution that evaluates ResNet architectures. Thus, we also evaluate the same configurations as COINN. This includes the smaller MiniONN architecture, a ResNet-32, and a ResNet-110.  We exclude GForce from this evaluation as they only evaluate VGG-16 models.

\paragraph{Inference Time.}
We again swap all MaxPool layers for AvgPool in our work, and \cryptgpu~but leave COINN unmodified.
We give the results in Figure~\ref{fig:coinn}.
We observe that, over each of the three increasingly large architectures, the trends are similar and proportional to the number of parameters (0.2, 0.5, and 1.7 million parameters for MiniONN, ResNet-32, and ResNet-110, respectively). Across all architectures and network delays, our work outperforms COINN by a statistically significant amount ($18\times$ on average in WAN).
We once again outperform \cryptgpu~in all evaluations\footnote{All of which are statistically significant except MiniONN in LAN where the confidence intervals overlap slightly.} with a $4\times$ speed up on average in the WAN.

\paragraph{Encrypted Accuracy.}
We give the results in Table~\ref{tab:coinn}.
We observe that \polytrick~is competitive with related work in all models, although we remark that, once again, our ResNet-18 models outperform all others. We also note that, while COINN does quantization-aware training, \polytrick~does not and still only loses a small amount of accuracy in encryption vs. plaintext.
\begin{table}[H]
    \begin{tabular}{llll}
        \toprule
    Dataset/Model                          & Technique  & Plain Acc     & Enc Acc        \\ \midrule
    \multirow{3}{*}{\begin{tabular}[c]{@{}l@{}}CIFAR-10 / \\ MiniONN\end{tabular}}   & \polytrick & 88.1 $\pm$ 0.26 & 87.9 $\pm$ 0.46 \\
                                           & \cryptgpu  & 91.2 $\pm$ 0.17 & 91.2 $\pm$ 0.16 \\
                                           & \coinn     & -             & 87.6              \\ \cmidrule(lr){1-4}
    \multirow{3}{*}{\begin{tabular}[c]{@{}l@{}}CIFAR-10 / \\ ResNet-110\end{tabular}}& \polytrick & 91.4 $\pm$ 0.18 & 91.4 $\pm$ 0.25  \\
                                           & \cryptgpu  & 92.8 $\pm$ 0.27 & 92.7 $\pm$ 0.26 \\
                                           & \coinn     & -             & 93.4              \\ \cmidrule(lr){1-4}
    \multirow{3}{*}{\begin{tabular}[c]{@{}l@{}}CIFAR-100 / \\ ResNet-32\end{tabular}} & \polytrick & 67.8 $\pm$ 0.32 & 67.8 $\pm$ 0.47 \\
                                           & \cryptgpu  & 68.4 $\pm$ 0.46 & 68.5 $\pm$ 0.45 \\
                                           & \coinn     & -             & 68.1             
    \\ \bottomrule
    \end{tabular}
    \caption{Accuracy comparison on the various architectures considered in \coinn~(5 runs).}\label{tab:coinn}
    \end{table}

    \subsection{Scaling to ImageNet}\label{sec:imagenet}
    In this section, we evaluate the scalability of our approach on the ImageNet dataset using a ResNet-50 architecture with 23 million parameters. This architecture was previously too large for training with polynomial activation functions~\cite{garimellaSisyphusCautionaryTale2021}. 
    We compare our approach to \cheetah, \coinn~and \cryptgpu and exclude GForce as they do not consider ImageNet. 
    
    \paragraph{Inference Time.}
    We plot the inference time in Figure~\ref{fig:imagenet}. We observe a significant reduction over \cheetah~and \coinn~across all network delays. We outperform \cheetah~by $39\times$ in the LAN (0.25 ms) and $15\times$ in the WAN (100 ms). Over \coinn, we observe a $28\times$ reduction in the LAN and a $90\times$ reduction in the WAN. Compared to \cryptgpu~we find that \polytrick~+ \honeytrick~is the fastest in all network delays by $3\times$ on average. \polytrick~+ \binotrick~is slightly slower in the LAN, but once again outperforms \cryptgpu~in the WAN.
    \begin{figure}
     \centering
     \includegraphics[width=0.7\columnwidth]{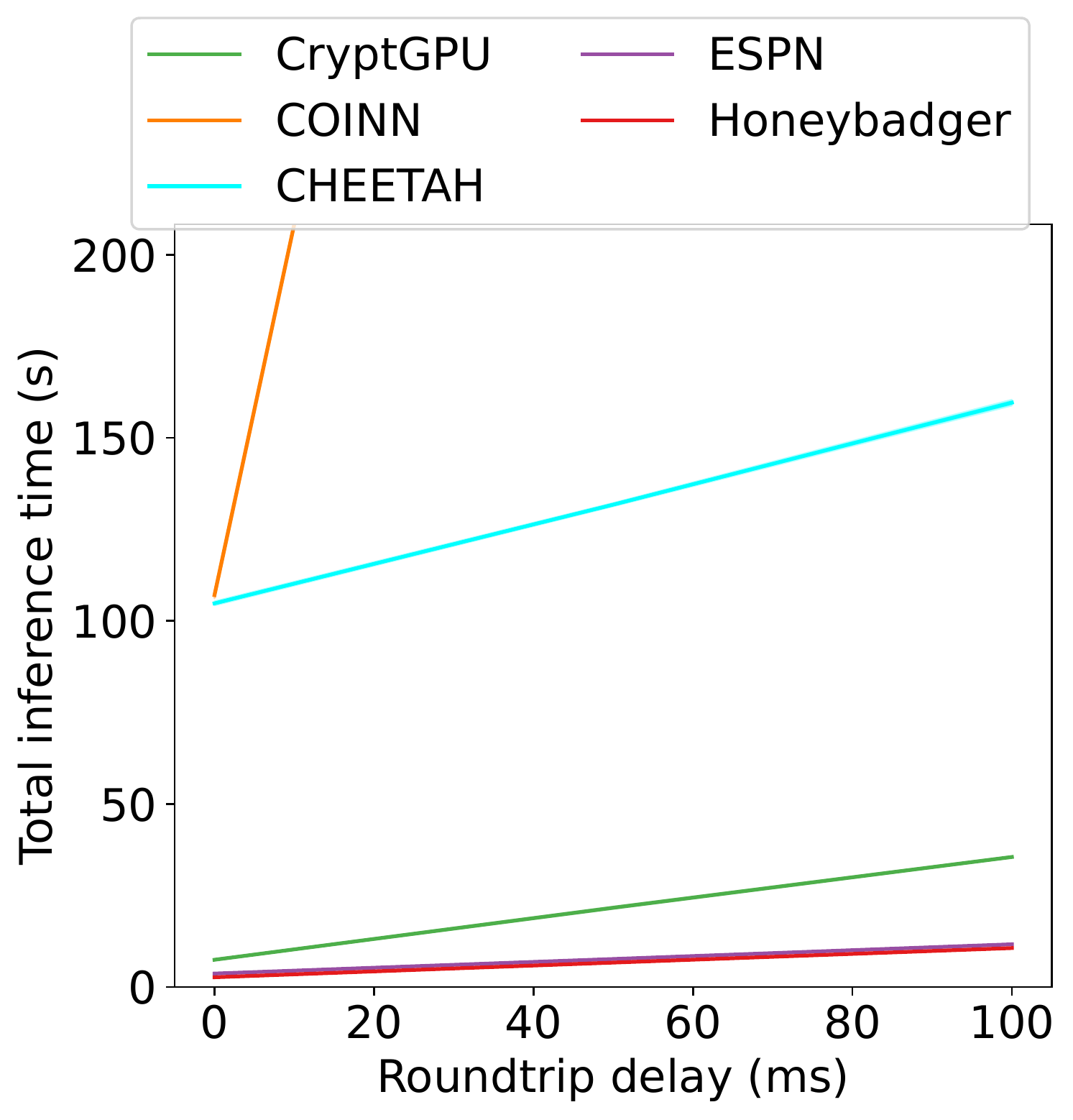}
     \caption{ImageNet evaluation on ResNet-50 (20 runs).}
     \label{fig:imagenet}
    \end{figure}
    \paragraph{Encrypted Accuracy.}
    We present a summary of the accuracies in Table~\ref{tab:imagenet}. We note that \cheetah~\cite{Haung_Cheetah_22} did not measure accuracy in either their code or paper so we omit them from this comparison. We observe a much higher encrypted accuracy for \polytrick~compared to \coinn~and thus, our solution is Pareto dominant. For \cryptgpu, we use a pre-trained PyTorch model with state-of-the-art accuracy. Therefore, as expected, \cryptgpu~has an accuracy $3\%$ higher than the model we trained from scratch. We note that with a higher degree polynomial, we were able to train a $79.2\%$ polynomial model. However, this model is not possible to infer in the $64$-bit field used by CrypTen (as higher degrees need more precision by equation~\ref{eq:precision}). We discuss future directions to further improve this result in Section~\ref{sec:discuss}.
    \begin{table}[H]
     \centering
     \begin{tabular}{lll}
     \toprule
     Technique & Plain Acc & Enc Acc \\ \midrule
     \polytrick & 77.7 & 77.3 \\
     \cryptgpu & 80.8 & 80.8 \\
     \coinn & - & 73.9 
     \\ \bottomrule
     \end{tabular}
     \caption{ImageNet accuracy comparison (1 run).}\label{tab:imagenet}
    \end{table}

\subsection{Real WAN Evaluation}\label{sec:real_wan}
\begin{table}
\centering
\begin{adjustbox}{width=\columnwidth}

\begin{tabular}{llll}
\toprule
Dataset / Model                                                       & \binotrick & \honeytrick & \cryptgpu \\ \midrule
\begin{tabular}[c]{@{}l@{}}CIFAR-10 / ResNet-110\end{tabular} & 49.3 $\pm$ 0.1               & 49.0 $\pm$ 0.1                & 242.5 $\pm$ 0.9             \\ \cmidrule(lr){1-4}
\begin{tabular}[c]{@{}l@{}}CIFAR-10 / ResNet-18\end{tabular}  & 15.3 $\pm$ 0.2               & 12.9 $\pm$ 0.1                & 48.9 $\pm$ 0.1              \\ \cmidrule(lr){1-4}
\begin{tabular}[c]{@{}l@{}}CIFAR-100 / ResNet-18\end{tabular} & 15.4 $\pm$ 0.2               & 12.9 $\pm$ 0.1                & 48.9 $\pm$ 0.5              \\ \cmidrule(lr){1-4}
\begin{tabular}[c]{@{}l@{}}CIFAR-100 / ResNet-32\end{tabular} & 14.1 $\pm$ 0.1               & 14.0 $\pm$ 0.1                & 75.4 $\pm$ 1.6              \\ \cmidrule(lr){1-4}
\begin{tabular}[c]{@{}l@{}}ImageNet / ResNet-50\end{tabular}  & 153.9 $\pm$ 1.0              & 104.9 $\pm$ 0.8               & 268.67 $\pm$ 1.3            \\ \bottomrule
\end{tabular}
\end{adjustbox}
\caption{Real World WAN evaluation. We report total inference time in seconds. }
\label{tab:amazon}
\end{table}
This section gives benchmarks for \binotrick, \honeytrick, and CryptGPU~\cite{tanCryptGPUFastPrivacyPreserving2021} in a real WAN. 
We use two AWS EC2 \texttt{g4dn.metal} instances, one in the Ohio data centre and one in Frankfurt, Germany. Each machine has $96$ cores, $384$ GB of memory, $100$ GB/s network bandwidth, and a NVIDIA T4 GPU. In practice, we measured 10.8MB/s bandwidth between the two instances. We run all the ResNet architectures and summarize the results in Table~\ref{tab:amazon}.
We repeat each experiment 20 times and report the mean and 95\% confidence intervals.
We observe similar trends to the simulated WAN used in previous experiments. 
Namely, \binotrick~and \honeytrick~perform similarly in runtime for most models. 
The exception is larger models like ResNet50, where the bandwidth limits impact \binotrick~more than \honeytrick.
In all cases, \cryptgpu~is significantly outperformed by both approaches.

\section{Discussion}\label{sec:discuss}
Our experimental evaluation in Section~\ref{sec:experiments} showed our algorithms significantly outperform all related work in WAN inference time. While state-of-the-art compared to other polynomial training approaches, \polytrick~still incurs a minor accuracy degradation compared to standard models with ReLUs. We posit a few directions for future work to further close this gap between polynomials and ReLUs.

\paragraph{Quantization.}
Note that aside from our quantization-aware polynomial fitting described in Section~\ref{sec:ml}, we have made no efforts to reduce the effects of quantization.
\coinn~developed training algorithms to help the model be robust to the overflow and quantization present in MPC~\cite{hussainCOINNCryptoML2021}.
An interesting future work would be to combine the \coinn~methods with \polytrick~to see if further accuracy gains are possible.

\paragraph{Precision.}
By using CrypTen as our backend, we were limited to a $64$-bit ring for cryptographic operations. As discussed in Section~\ref{sec.poly_eval}, this precision determines the degree and range of polynomials we can use (due to either the failure probability of truncation or severe truncation of intermediate values). Interesting future work is to increase this precision to enable higher-degree polynomials and study the performance-accuracy trade-off. 
Our initial results on ImageNet show that we can train up to a degree eight polynomial without suffering escaping activations. However, we could not increase the ring size to study the effect of higher degrees on inference time.

\paragraph{MaxPools.}
We recall that a MaxPool layer requires comparisons and, thus, expensive conversions to binary shares (like ReLUs). Therefore, we replaced all MaxPools with AvgPool layers. However, in some architectures, such as VGG-16~\cite{simonyanVeryDeepConvolutional2015}, we found that swapping MaxPool for AvgPool degraded accuracy by up to $6\%$. Finding an efficient MaxPool alternative for architectures like VGG is important for future work. However, since the ResNet models give high accuracy using AvgPool layers we did not pursue this issue further.

\section {Related Work}\label{sec:related_work}
This work focuses on achieving state-of-the-art run time and accuracy in two-party secure inference.
We measure this objective by evaluating against the current state-of-the-art as determined by a recent SoK by Ng and Chow~\cite{ngSoKCryptographicNeuralNetwork2023}.
Namely, we compare to COINN~\cite{hussainCOINNCryptoML2021}, GForce~\cite{ngGForceGPUFriendlyOblivious2021} and CrypTen~\cite{knottCrypTenSecureMultiParty2021,tanCryptGPUFastPrivacyPreserving2021} in Section~\ref{sec:experiments} as they represent the Pareto front according to Ng and Chow~\cite{ngSoKCryptographicNeuralNetwork2023}.
Another potential candidate on the Pareto front is Falcon~\cite{liFALCONFourierTransform2020}, with low latency and accuracy~\cite{liFALCONFourierTransform2020}.
We did not evaluate Falcon as the accuracy drop was too significant (over 10\%~\cite{ngSoKCryptographicNeuralNetwork2023}).
Furthermore, \gforce~is shown to outperform Falcon in both latency and accuracy, and we outperform GForce~\cite{ngGForceGPUFriendlyOblivious2021}.
For a complete list of other works not on the Pareto front, we defer to Ng and Chow's work~\cite{ngSoKCryptographicNeuralNetwork2023}.
Notably, many works consider different threat models or use different approaches, such as homomorphic encryption.
We leave extending our polynomial activation functions to these settings for future work.
For the remainder of this section, we discuss works with a similar approach to ours that are not state-of-the-art or not evaluated by Ng and Chow~\cite{ngSoKCryptographicNeuralNetwork2023}.

\paragraph{Replacing or Reducing ReLU's.}
It has been established that the non-linear functions such as ReLU are the bottleneck for secure computation~\cite{garimellaSisyphusCautionaryTale2021,hussainCOINNCryptoML2021,mishraDelphiCryptographicInference2020,CryptoNASProceedings34th}.
Several works initially focused on reducing the number of ReLU activations, optimizing for the best trade-off between accuracy and runtime~\cite{CryptoNASProceedings34th,jhaDeepReDuceReLUReduction2021}.
A faster approach is to replace all ReLU's entirely using polynomial approximations~\cite{garimellaSisyphusCautionaryTale2021}.
In Section~\ref{sec:ml}, we discussed the most recent work in this space, Sisyphus~\cite{garimellaSisyphusCautionaryTale2021}.
While making significant progress toward training models with polynomial activations, Sisyphus could not overcome the escaping activation problem for models with more than 11 layers. 
Before Sisyphus, there were a handful of works on smaller models that typically focus on partial replacement (some ReLU's remained)~\cite{gilad-bachrachCryptoNetsApplyingNeural2016,mishraDelphiCryptographicInference2020,mohasselSecureMLSystemScalable2017}. 
An interesting exception from Lee et al. used degree 29 polynomials in HE but suffered prohibitively high runtimes~\cite{leePreciseApproximationConvolutional2021}.
Our work is the first to make high-accuracy polynomial training feasible (without escaping activations) in deep neural networks.

A notable recent work is PolyKervNets~\cite{aremuPolyKervNetsActivationfreeNeural2023}.
Inspired by the computer vision literature, PolyKervNets remove the activation functions and instead exponentiate the output of each convolutional layer~\cite{aremuPolyKervNetsActivationfreeNeural2023}. The problem with this approach is that, similar to polynomial activation functions, the exponents make the training unstable. Aremu and Nandakumar note that exploding gradients prevent their approach from scaling to ResNet models deeper than ResNet18 (using degree 2 polynomials). Furthermore, PolyKervNets only allow for a single fully connected layer which reduces the accuracy of the models. Conversely, \polytrick~scales to deeper models such as ResNet110 and much higher degrees. Moreover, we achieve significantly better plaintext accuracy on ResNet-18 (93.4 vs 90.1 on CIFAR-10 and 74.9 vs 71.3 on CIFAR-100).

\paragraph{Polynomial Evaluation in MPC.}
Our work focuses on co-designing the activation functions with cryptography by using polynomials. However, the problem of computing polynomials in MPC is of independent interest and has also been studied in the literature. The state-of-the-art in this space is HoneyBadger, as discussed in Section~\ref{sec:crypto}. Other notable works include the initial inspiration for HoneyBadger from Damg\aa{}rd et al.~\cite{damgardUnconditionallySecureConstantRounds2006}.
This early approach conducts exponentiation by blinding and reconstructing the number to be exponentiated so the powers can be computed in plaintext~\cite{damgardUnconditionallySecureConstantRounds2006}. Building off this idea, Polymath constructs a constant round protocol for evaluating polynomials focused on matrices~\cite{luPolymathLowLatencyMPC2022}. However, HoneyBadger outperforms Polymath by reducing both the rounds and the number of reconstructions to one.

\section{Conclusion}
In this work, we co-designed the ML and MPC aspects of secure inference to remove the bottleneck of non-linear layers.
\polytrick~maintains a competitive inference accuracy while being significantly faster in wide area networks using novel single round MPC protocols (\binotrick~and \honeytrick).
Our state-of-the-art inference times motivate future work to further improve the ML accuracy of polynomial activations in DNNs.
\section*{Acknowledgements}
We gratefully acknowledge the support of the Natural Sciences and Engineering Research Council (NSERC) for grants RGPIN-05849, and IRC-537591, the Royal Bank of Canada, and Amazon Web Services Canada.

\section*{Availability}

We make all source code to reproduce our experiments available here: 
\url{https://github.com/LucasFenaux/PILLAR-ESPN}.

\bibliographystyle{plain}
\bibliography{main.bib}
\appendix
\section{Evaluation of BatchNorm with Polynomials}\label{app:batchnorm}
\begin{table}[H]
\begin{tabular}{lll}
\toprule
            & Without BatchNorm & With BatchNorm \\ \midrule
Normal Relu & 88.77 $\pm$ 0.11           & 90.05 $\pm$ 0.16         \\
PolyRelu    & 82.99 $\pm$ 0.42           & 87.37  $\pm$ 0.13        \\ \bottomrule
\end{tabular}\caption{Comparing the effect of BatchNorm on MiniONN model.}\label{tab:batchnorm}
\end{table}

To study the effect of a batch norm layer on our training process, we train a standard ReLU model and a model with polynomial activations both with and without batch norm layers. 
We use the MiniONN architecture and give the results averaged over three random seeds with $95\%$ confidence intervals in Table~\ref{tab:batchnorm}. 
We find that batch norm layers improve both models. 
However, the improvement due to using batch norm is significantly greater when using polynomial activation functions. 
This is an intuitive result as batch norm helps keep each layer's output bounded and thus reduces the work of our regularization function.

\section{Evaluation of Sigmoid with Polynomials}\label{app:sigmoid}
\begin{table}[H]
    \centering
\begin{tabular}{lll}
\toprule
            & ReLU & Sigmoid \\ \midrule
Standard Activation & 94.7 $\pm$ 0.08           & 90.2 $\pm$ 0.11         \\
Polynomial Approx.   & 93.4 $\pm$ 0.16           & 85.9  $\pm$ 0.11        \\ \bottomrule
\end{tabular}\caption{Comparing the effect of the activation function on a ResNet18 model.}\label{tab:sigmoid}
\end{table}

In this work, we focus on the ReLU activation function, the default in common architectures such as ResNets~\cite{heDeepResidualLearning2016}.
Another reason we focus on ReLU is that it gives better accuracy than alternatives like Sigmoid.
In this section, we highlight this accuracy advantage by the accuracy of a ResNet18 model with different activation functions.
In Table~\ref{tab:sigmoid}, we evaluate both ReLU and Sigmoid with and without using polynomial approximation.
In all cases, we use the ResNet18 architecture with default parameters given in Section~\ref{sec:exp_setup}, we note the polynomial is of degree  $\degree=4$.
The results are averaged over five random seeds and shown with $95\%$ confidence intervals. 
We find that ReLU consistently outperforms Sigmoid with and without using polynomial evaluations.
However, the accuracy of Sigmoid decreases more, relative to ReLU, when using polynomial approximation.
This result further motivates our use of ReLU.
We leave further investigation of Sigmoid and other activations for future work. 

\section{Quantization Aware Polynomial Fitting}\label{app:poly_fit}
\begin{figure}[H]
    \centering
\includegraphics[width=0.75\columnwidth]{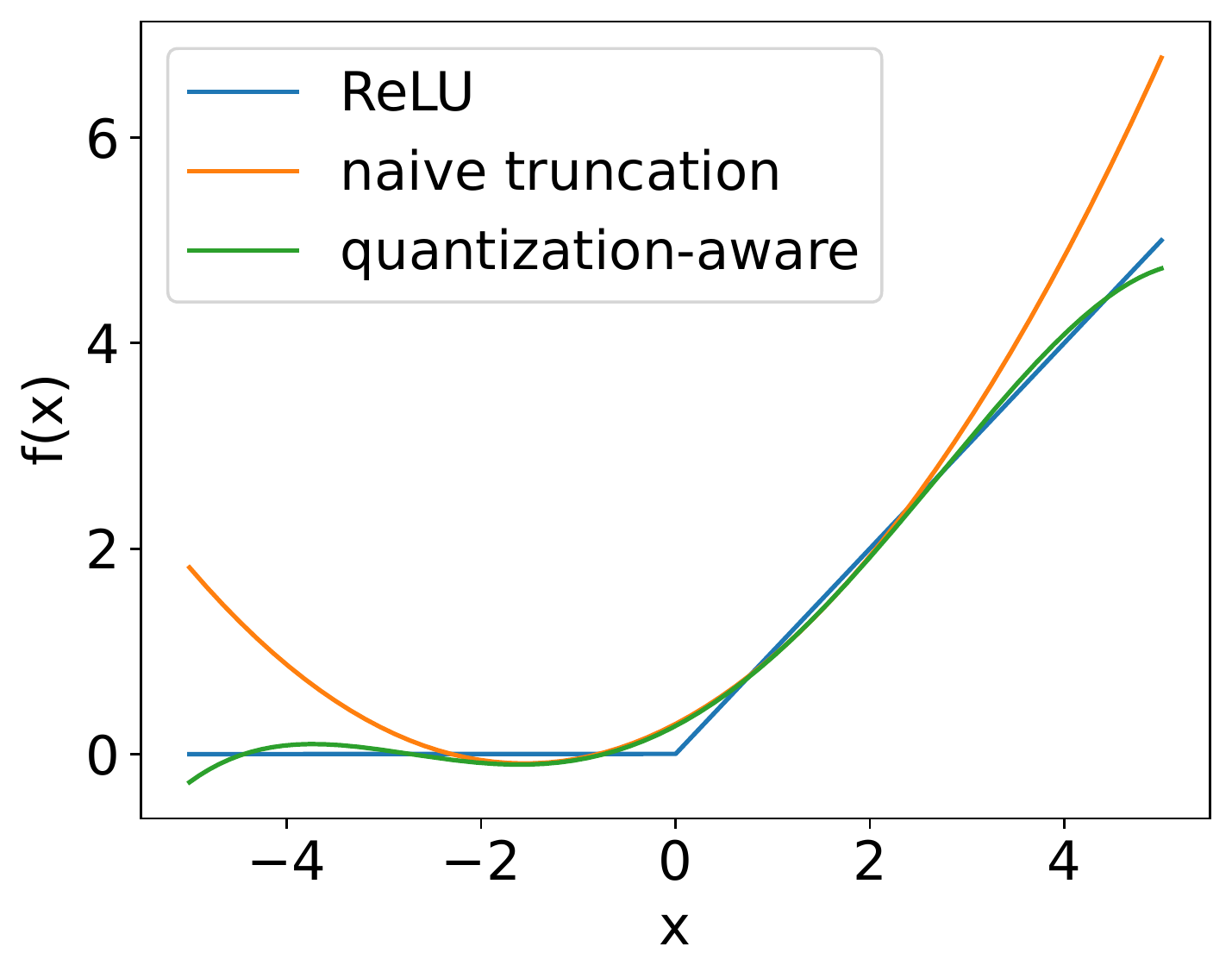}
    \caption{The effect of truncation on a polynomial activation function. }
    \label{fig:quant}
\end{figure}
In Figure~\ref{fig:quant}, we plot the polynomial approximation with and without our quantize-aware fitting approach described in Section~\ref{sec:poly_blow_up}.
First, we plot the polynomial approximation after truncation and see that it diverges from a true ReLU.
We also plot our quantized polynomial fitting and show that it addresses the problems of exploding activations within the range. 

\paragraph{Polynomial Input Distribution After Regularization}
We plot the histogram of the input to all activation functions of a ResNet 18 model on CIFAR-10 in Figure~\ref{fig:laplace}.
\begin{figure}
    \centering
    \includegraphics[width=0.85\columnwidth]{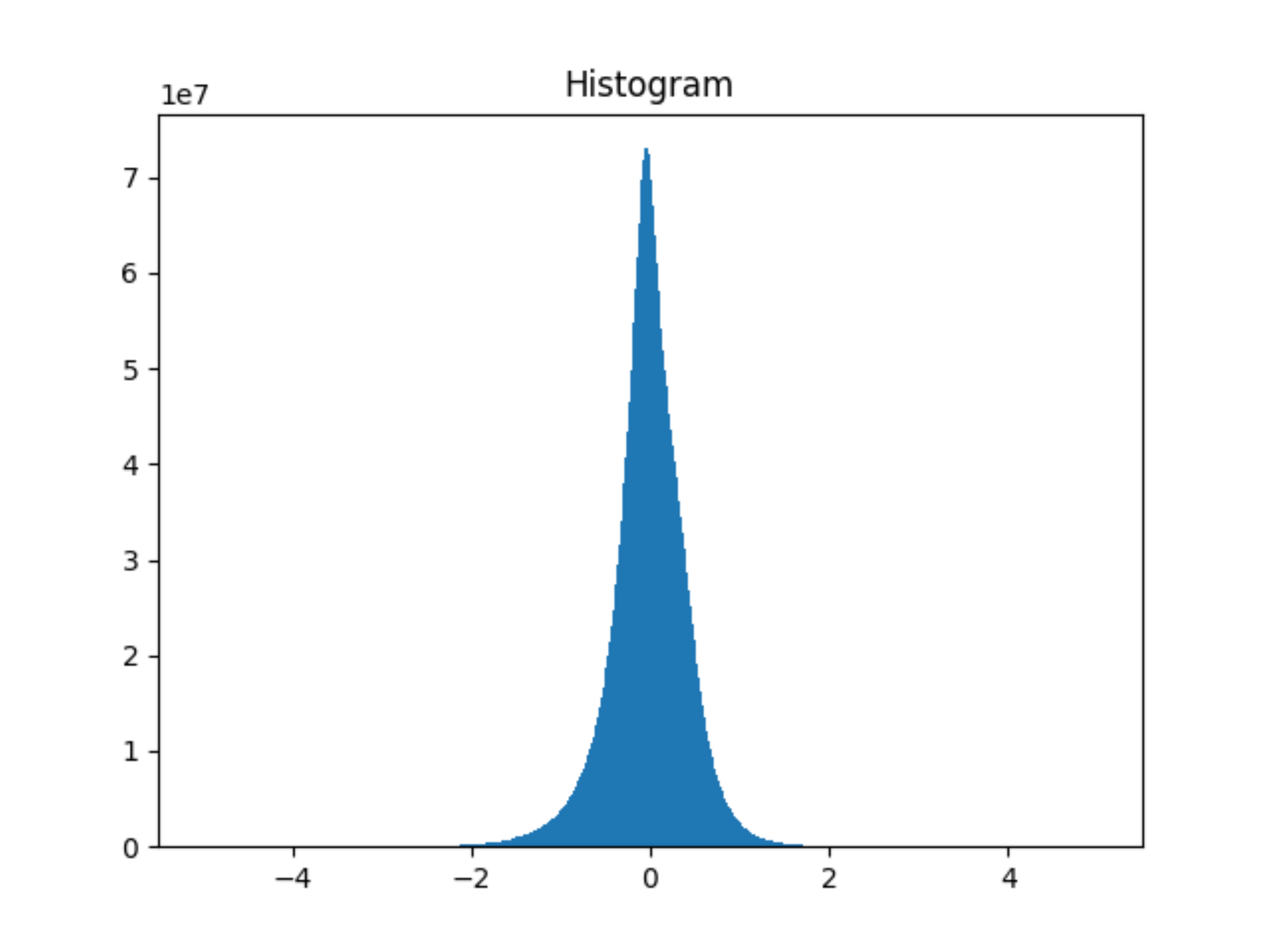}
    \caption{Polynomial Input Distribution of ResNet18 on CIFAR10.}
    \label{fig:laplace}
\end{figure}
\section{Communication and Rounds Benchmark}\label{app:comm_rounds}
\begin{table*}[]
\centering
\begin{tabular}{cccccccc}
\toprule
Dataset                   & Model     & \multicolumn{1}{l}{\binotrick} & \honeytrick & \cryptgpu    & \coinn & \multicolumn{1}{l}{\gforce} & \multicolumn{1}{l}{Cheetah} \\ \midrule
\multirow{3}{*}{CIFAR-10}  & ResNet-18  & 0.46 (38)                & 0.23 (38)   & 0.45 (174)  & /     & /                          & /                           \\
                          & ResNet-110 & 0.55 (221)               & 0.13 (221)  & 0.53 (1093) & 6.8   & /                          & /                           \\
                          & VGG-16     & 0.32 (31)                & 0.26 (31)   & /           & /     & 0.050                      & /                           \\ \cmidrule(lr){1-8}
\multirow{3}{*}{CIFAR-100} & ResNet-18  & 0.46 (38)                & 0.23 (38)   & 0.45 (174)  & /     & /                          & /                           \\
                          & ResNet-32  & 0.16 (65)                & 0.037 (65)  & 0.15 (313)  & 1.9   & /                          & /                           \\
                          & VGG-16     & 0.32 (31)                & 0.26 (31)   & /           & /     & 0.050                      & /                           \\ \cmidrule(lr){1-8}
ImageNet                  & ResNet-50  & 7.85 (160)               & 4.04 (160)  & 7.70 (552)  & 122.0 & /                          & 2.36 (1042)  \\ \bottomrule           
\end{tabular}
\caption{Evaluation of communication in GB and the number of rounds (shown in parenthesis).}\label{tab:comm_rounds}
\end{table*}
We study two additional evaluation metrics of rounds and communication in this section.
Recall that the number of rounds significantly affects the protocol latency over WAN.
The communication affects each round's throughput, depending on the network bandwidth.
In Table~\ref{tab:comm_rounds}, we summarize the communication in GB and the number of rounds (shown in parenthesis) for our work and the related work we compare to in Section~\ref{sec:experiments}. We note that both \coinn~and \gforce~do not log the communication or rounds in their code base. Thus, we report their communication numbers from the corresponding tables in the papers (\coinn~\cite[Table 3]{hussainCOINNCryptoML2021} and \gforce~\cite[Table 7]{ngGForceGPUFriendlyOblivious2021}). Neither work evaluates the number of rounds.

In all cases, \binotrick~and \honeytrick~significantly dominate all evaluated related work in the number of rounds with a $3-5\times$ improvement.
The communication of \binotrick~and \honeytrick~is significantly less than \coinn, approximately the same as \cryptgpu, and more than \gforce~and \cheetah.
Our \honeytrick~solution has approximately $2\times$ the communication of \cheetah, but \cheetah~has $10\times$ the rounds; thus, in practice, \honeytrick~gives much better performance as shown in Section~\ref{sec:experiments}. 
\gforce~has an impressively low communication online (50MB) due to offloading 20GB of communication to an offline phase.
Once again, despite having higher communication, we recall that both \binotrick~and \honeytrick~dominate related work in runtime in the WAN (as shown in Section~\ref{sec:experiments}).

\section{Correctness Proofs}\label{app:correct_proofs}
We begin by proving the correctness of Algorithm~\ref{alg:binomial_exp}.
\begin{thm}\label{thm:alg1}
    Given an input $\shares{\base}$ = $\base_\partyA + \base_\partyB$ and exponent $\exponent$, Algorithm~\ref{alg:binomial_exp}, correctly returns $\shares{\base^\exponent}$.
\end{thm}
\begin{proof}
    We begin with party~\partyA. In line~\ref{line:partya}, they compute their share of the vector $\mathbf{\partyaVec}$ where $\partyaVec_i = \base_\partyA^{\exponent-i}$
    (since $\mathbf{\partyaVec}$ was initialized to zero and the for loop iterates over each entry of $\mathbf{\partyaVec}$ exactly once). Similarly, in line~\ref{line:partyb}, party~\partyB~computes their share of $\mathbf{\partybVec}$ where $\partybVec_i = {\exponent \choose i} \base_\partyB^i$.
    We recall that party \partyB's share of $\partyaVec$ is the zero vector, and similarly for party \partyA's share of $\partybVec$
    Then, the vector $\mathbf{p}$ is obtained by multiplying $\partyaVec$ and $\partybVec$ element wise in line~\ref{line:mult_shares}. Thus, $\mathbf{p}_i = \partyaVec_i \cdot \partybVec_i = \base_\partyA^{\exponent-i} {\exponent \choose i} \base_\partyB^i$.
    The final step (line~\ref{line:get_result}), simply sums $\mathbf{p}$.
    Therefore, 
    \begin{equation}
        s = \sum_{i=1}^\exponent \mathbf{p_i} = \sum_{i=1}^\exponent {\exponent \choose i} \base_\partyA^{\exponent-i}\base_\partyB^i
    \end{equation}
    which applying the binomial theorem (\ref{eq:bino_formula}) gives $(\base_\partyA + \base_\partyB)^\exponent = \shares{\base^\exponent}$.
\end{proof}

Given the correctness of Algorithm~\ref{alg:binomial_exp}, we now prove the correctness of Algorithm~\ref{alg:poly_eval}. To begin, we bound the failure probability of each truncation step in Algorithm~\ref{alg:poly_eval}.

\begin{thm}\label{thm:first_trunc}
    Consider computing a degree $\degree$ polynomial fitted to the range $[-\range, \range)$. 
    Let the global precision (size of the ring) be $\totalprecision$-bit and the working precision of each value be $\precision$-bit.
    Then, the truncation in line~\ref{line:pre_scale_down} of Algorithm~\ref{alg:poly_eval}, fails (The local division of $[[\base]_\partyA, [\base]_\partyB]* 2^{-\scaleDown}\neq \shares{\base * 2^{-\scaleDown}}$) with probability at most
\begin{equation}\label{eq:precision}
    Pr[\text{Line~\ref{line:pre_scale_down} Failure}] \leq \frac{2^{\lceil \log_2{\range}\rceil+1 + \precision}}{2^{\totalprecision}}.
\end{equation}
\end{thm}
\begin{proof}
    Consider the first truncation of Algorithm~\ref{alg:poly_eval} in line~\ref{line:pre_scale_down}.
    The input to this truncation is $\base$ which we assume is contained in the range $[-\range, \range)$ with a working precision of $\precision$-bits.
    Therefore, the size of $\base$ is $2^{\lceil \log_2{\range}\rceil+1}$ in the integer part and $2^\precision$ in the decimal part.
    Which gives, $|x| \leq 2^{\lceil \log_2{\range}\rceil+1 + \precision}$.
    Given that we are working in a $\totalprecision$-bit ring and the probability of failure of the truncation protocol is bounded by $|x|/Q$ where $Q$ is the ring size~\cite{knottCrypTenSecureMultiParty2021}, the result follows.
\end{proof}
\begin{thm}\label{thm:second_trunc}
    Consider computing a degree $\degree$ polynomial fitted to the range $[-\range, \range)$. 
    Let the global precision (size of the ring) be $\totalprecision$-bit and the working precision of each value be $\precision$-bit.
    Then, the truncation in line~\ref{line:post_scale_down} of Algorithm~\ref{alg:poly_eval} fails (The local division of $[[\base]_\partyA, [\base]_\partyB]* 2^{-\precision*(i-1)+ \scaleDown\cdot i}\neq \shares{\base * 2^{-\precision*(i-1)+ \scaleDown\cdot i}}$) with probability at most
\begin{equation}\label{eq:precision}
    Pr[\text{Line~\ref{line:post_scale_down} Failure}] \leq \frac{2^{i(\lceil \log_2{\range}\rceil+1) + 2\precision}}{2^{\totalprecision}}.
\end{equation}
where $i$ is the power of $\base$ being truncated ($i\leq \degree$).
\end{thm}
\begin{proof}
    Consider the truncation in line~\ref{line:post_scale_down} of Algorithm~\ref{alg:poly_eval}. 
    The input to this truncation is the output of the previous truncation in line~\ref{line:pre_scale_down}, raised to the power $i$.
    We assume the previous truncation was correct.
    Then after the truncation, 
    \begin{equation}
        |x'_i| = \frac{2^{\lceil \log_2{\range}\rceil+1 + \precision}}{2^{\lceil (i-2) \precision / i \rceil}} \leq \frac{2^{\lceil \log_2{\range}\rceil+1 + \precision}}{ 2^{(i-2) \precision / i} }
    \end{equation}
    for $i\in\{2,\dots, \degree\}$, where the inequality holds because $(i-2) \precision / i$ is positive.
    After applying the exponentiation by $i$ we get $|\mathbf{p}_i| \leq 2^{i(\lceil \log_2{\range}\rceil+1) + 2\precision}$.
    Given that we are working in a $\totalprecision$-bit ring and the probability of failure of the truncation protocol is bounded by $|x|/Q$ where $Q$ is the ring size~\cite{knottCrypTenSecureMultiParty2021}, the result follows.
\end{proof}

\begin{thm}\label{thm:alg2}
Given an input $\shares{\base}$ = $\base_\partyA + \base_\partyB$ and polynomial coefficients $\coefficients$, Algorithm~\ref{alg:poly_eval}, correctly returns the polynomial evaluation $\shares{\sum_{i=0}^\degree \coefficients_i \cdot \base^i}$ except with probability.
\begin{equation*}
    Pr[\text{Alogrithm~\ref{alg:poly_eval} Fails}] \leq 
     \frac{\sum\limits_{i=2}^\degree\left( 2^{i(\lceil \log_2{\range}\rceil+1) + 2\precision} + 2^{\lceil \log_2{\range}\rceil+1 + \precision}\right)}{2^{\totalprecision}}.
\end{equation*}
\end{thm}
\begin{proof}
    First, prove correctness assuming the truncation operators are correct (then we will account for the probability of failure).
    We begin by proving that $\mathbf{p'_i}=x^i$ for $i \in \{1,\dots, \degree\}$ after line~\ref{line:post_scale_down}.
    We note that $\base$ is actually $\base * 2^\precision$ due to the fixed point encoding.
    $\mathbf{p_1}$ follows trivially.
    For $i \in \{2,\dots, \degree\}$, we work backwards, from line~\ref{line:post_scale_down} to line~\ref{line:scale_down_factor} expanding the definition of $\mathbf{p'_i}$
    \begin{eqnarray}
        \mathbf{p'_i} &=& \mathbf{p}_i * 2^{-\precision*(i-1)+ \scaleDown\cdot i} \\
        &=& (x'_i)^i * 2^{-\precision*(i-1)+ \scaleDown\cdot i} \\
        &=& (x \cdot 2^\precision \cdot 2^{-\scaleDown})^i * 2^{-\precision*(i-1)+ \scaleDown\cdot i}\label{eq:substitution}\\
        &=& x^i \cdot 2^\precision
    \end{eqnarray}
    where the first lines, up to (\ref{eq:substitution}), come from the substitution of lines~\ref{line:post_scale_down} through~\ref{line:pre_scale_down}, respectively. The remaining step follows from basic algebra.

    Finally, we bound the failure probability of the algorithm. The truncation steps in line~\ref{line:pre_scale_down} and line~\ref{line:post_scale_down} each introduce a possibility of wrap around error. Each truncation is executed $\degree-1$ times. Thus, applying the bounds derived in Theorem~\ref{thm:first_trunc} and Theorem~\ref{thm:second_trunc} for each $i$ in the for loop, the total failure probability follows.
\end{proof}

\end{document}